\documentclass[10pt,usletter]{article}

\usepackage[totalwidth=450pt,totalheight=640pt]{geometry}

\usepackage{graphicx}
\usepackage{lmodern}
\usepackage[utf8]{inputenc}
\usepackage{fullpage}
\usepackage{amsmath}
\usepackage{hyperref}
\usepackage{csquotes}
\usepackage{microtype}
\usepackage{mathtools}
\usepackage{amssymb}
\usepackage{booktabs}
\usepackage{multirow}
\usepackage{wrapfig}
\usepackage[margin,noinline,draft]{fixme}
\usepackage[sort&compress,numbers,sectionbib]{natbib}
\usepackage{cryptocode}
\usepackage{tcolorbox}
\usepackage{tikz}
\usetikzlibrary{matrix}
\usetikzlibrary{graphs}
\usetikzlibrary{positioning, arrows.meta, calc}
\usepackage{subcaption}
\usepackage{amsthm}

\newtheorem{theorem}{Theorem} 
\newtheorem{proposition}{Proposition} 
\newtheorem{corollary}{Corollary} 
 
\newtheorem{lemma}{Lemma}
\newtheorem{remark}{Remark}
\newtheorem{definition}{Definition} 

\makeatletter
\newtheorem*{rep@theorem}{\rep@title}
\newcommand{\newreptheorem}[2]{%
\newenvironment{rep#1}[1]{%
 \def\rep@title{#2 \ref{##1}}%
 \begin{rep@theorem}}%
 {\end{rep@theorem}}}
\makeatother

\newreptheorem{theorem}{Theorem}
\newreptheorem{lemma}{Lemma}

\def\ve#1{\mathchoice{\mbox{\boldmath$\displaystyle\bf#1$}}
{\mbox{\boldmath$\textstyle\bf#1$}}
{\mbox{\boldmath$\scriptstyle\bf#1$}}
{\mbox{\boldmath$\scriptscriptstyle\bf#1$}}}

\newcommand\veb{{\ve b}}

\newcommand\veg{{\ve g}}
\newcommand\veh{{\ve h}}

\newcommand\vel{{\ve l}}

\newcommand\veo{{\ve o}}
\newcommand\vep{{\ve p}}
\newcommand\veq{{\ve q}}
\newcommand\ver{{\ve r}}
\newcommand\ves{{\ve s}}

\newcommand\veu{{\ve u}}
\newcommand\vev{{\ve v}}
\newcommand\vew{{\ve w}}
\newcommand\vex{{\ve x}}
\newcommand\vey{{\ve y}}
\newcommand\vez{{\ve z}}
\newcommand\vealpha{{\ve \alpha}}
\newcommand\vebeta{{\ve \beta}}

\newcommand\verho{{\ve \rho}}
\newcommand\vezero{{\ve 0}}

\def\R{\mathbb{R}}
\def\Z{\mathbb{Z}}
\def\N{\mathbb{N}}
\def\G{\mathcal{G}}

\def\wt{\mathbf{wt}}

\newcommand{\suppo}{\textrm{supp}}

\newcommand{\CC}{\mathcal{C}}

\newcommand{\X}{\mathbb{X}}

\makeatletter
\providecommand{\vrect}{\mkern1mu\mathpalette\v@rectangle\relax\mkern1mu}
\newcommand{\v@rectangle}[2]{%
  \hbox{%begingroup
  \fboxrule=0.5\fontdimen 8
    \ifx#1\displaystyle\textfont\else
    \ifx#1\textstyle\textfont\else
    \ifx#1\scriptstyle\scriptfont\else
    \scriptscriptfont\fi\fi\fi 3
  \fboxsep=-\fboxrule
  \fbox{$\m@th#1\phantom{(}$}%
  }%\endgroup
}
\makeatother

\DeclarePairedDelimiter\ceil{\lceil}{\rceil}
\DeclarePairedDelimiter\floor{\lfloor}{\rfloor}

\DeclareMathOperator{\sign}{sign}

\usepackage{stackengine}
\stackMath

\newcommand{\df}{:=}

\usetikzlibrary{fit,matrix,patterns,decorations.pathreplacing}

\newcommand{\Oh}{\mathcal{O}}
\newcommand{\OhOp}[1]{\Oh\mathopen{}\mathclose\bgroup\left( #1 \aftergroup\egroup\right)}
\DeclareMathOperator{\poly}{{\rm poly}}

\usepackage{xspace}
\newcommand{\FPT}{{\sf FPT}\xspace}

\newcommand{\NPh}{\hbox{{\sf NP}-hard}\xspace}
\newcommand{\NPhness}{\hbox{{\sf NP}-hardness}\xspace}

\newcommand{\XP}{{\sf XP}\xspace}
\newcommand{\W}[1]{{\sf W[#1]}\xspace}

\newcommand{\Whness}[1]{$\mathsf{W[#1]}$-hardness\xspace}

\usepackage{tabularx}
\newcommand{\prob}[3]{
\begin{center}
\begin{tabularx}{\textwidth}{lX}
	\multicolumn{2}{l}{#1}\\
	{\bf Input:}&{#2}\\
	{\bf Goal:}&{#3}
\end{tabularx}
\end{center}
}

\usepackage{etoolbox}
\patchcmd{\thebibliography}{\chapter*}{\section*}{}{}

\usepackage{xspace}

\newcommand{\ie}{i.\,e.\xspace}
\newcommand{\wrt}{w.\,r.\,t.\xspace}

\newcommand{\eg}{e.\,g.\xspace}

\newcommand{\dist}{\mathrm{dist}}

\newcommand{\ksum}{\textsc{Subset Sum}}

\begin{document}
\title{Sometimes, Convex Separable Optimization \emph{Is} Much Harder than Linear Optimization, and Other Surprises}

\author{Cornelius Brand$^1$\thanks{Supported by OP RDE project No. CZ.02.2.69/0.0/0.0/18 053/0016976 International mobility of research, technical and administrative staff at the Charles University, Prague and by the Austrian Science Fund (FWF) via projects Y1329
 (Parameterized Analysis in Artificial Intelligence) and P31336 (New Frontiers for Parameterized
 Complexity).}, Martin Kouteck{\'{y}}$^2$\thanks{Partially supported by Charles University project UNCE/SCI/004 and by the project 19-27871X of GA ČR.}, Alexandra Lassota$^3$\thanks{Supported by the Swiss National Science Foundation within the project \emph{Lattice algorithms and Integer Programming}~(200021\_185030/1).}, and Sebastian Ordyniak$^4$\thanks{Supported by the
Engineering and Physical Sciences Research Council (EPSRC, project EP/V00252X/1).}\\[0.1cm]
  \mbox{}\small$^1$Vienna University of Technology, Austria (\small cbrand@ac.tuwien.ac.at)\\
  \mbox{}\small$^2$Charles University, Czech Republic (\small koutecky@iuuk.mff.cuni.cz)\\
  \mbox{}\small$^3$EPFL, Lausanne, Switzerland (\small alexandra.lassota@epfl.ch)\\
  \mbox{}\small$^4$University of Leeds, UK (\small sordyniak@gmail.com)
}
  
% \author{Cornelius Brand\thanks{supported by OP RDE project No.
% CZ.02.2.69/0.0/0.0/18 053/0016976 International mobility of research, technical and administrative staff at the Charles University, Prague and by the Austrian Science Fund (FWF) via projects Y1329
% (Parameterized Analysis in Artificial Intelligence) and P31336 (New Frontiers for Parameterized
% Complexity).}\inst{1} \and
% Martin Kouteck{\'{y}} \thanks{partially supported by Charles University project UNCE/SCI/004 and by the project 19-27871X of GA ČR.}\inst{2}\and
% Alexandra Lassota \thanks{supported by project 200021\_185030/1, Lattice algorithms and Integer Programming} \inst{3}\and
% Sebastian Ordyniak \inst{4}}
% \institute{Department of Computer Science, Vienna University of Technology, Vienna, Austria\\
% \email{cbrand@ac.tuwien.ac.at}\\ \and
% Computer Science Institute, Charles University, Prague, Czech Republic\\
% \email{koutecky@iuuk.mff.cuni.cz}\\ \and
% Institute of Mathematics, EPFL, Lausanne, Switzerland\\
% \email{alexandra.lassota@epfl.ch} \\ \and
% School of Computing, University of Leeds, Leeds, United Kingdom\\
% \email{s.ordyniak@leeds.ac.uk}}
\date{}

\maketitle

\begin{abstract}
An influential 1990 paper of Hochbaum and Shanthikumar made it common wisdom that ``convex separable optimization is not much harder than linear optimization''~[JACM 1990].
%. This has since become ``common wisdom.''
We exhibit two fundamental classes of mixed integer (linear) programs that run counter this intuition.
Namely those whose constraint matrices have small coefficients and small primal or dual treedepth: While linear optimization is easy~[Brand, Koutecký, Ordyniak, AAAI 2021], we prove that separable convex optimization \emph{is} much harder.
Moreover, in the pure integer and mixed integer linear cases, these two classes 
%(bounded primal and dual treedepth) 
have the same parameterized complexity. We show that they yet behave quite differently in the separable convex mixed integer case.
%All of these results break popular intuitions and ``common wisdom''.

Our approach employs the mixed Graver basis introduced by Hemmecke~[Math. Prog. 2003]. %Surprisingly, it hasn't been studied since.
We give the first non-trivial lower and upper bounds on the norm of mixed Graver basis elements.
In previous works involving the integer Graver basis, such upper bounds have consistently resulted in efficient algorithms for integer programming.
Curiously, this does not happen in our case. In fact, we even rule out such an algorithm.

% \keywords{Separable Convex Optimization, Mixed Integer (Linear) Programs, Fixed-Parameter Tractability, Treedepth, (Mixed) Graver basis, Complexity Theory}
\end{abstract}

\section{Introduction}
%\todo{Martin: Suggest removing lb, ub from keywords.}
The \textsc{Mixed Integer Programming} problem is to solve 
\begin{equation}
	\min f(\vex):\, E\vex = \veb, \, \vel \leq \vex \leq \veu, \, \vex \in \Z^{n_\Z} \times \R^{n_{\R}},\label{MIP} \tag{MIP}
\end{equation}
where the number of columns is $n = n_{\Z} + n_{\R}$, the objective function is $f: \R^n \to \R$, $E \in \Z^{m \times n}$, $E$ denotes the constraint matrix, $\veb \in \R^m$ is the right-hand side, and the lower and upper bounds are $\vel, \veu \in (\R \cup \{\pm \infty\})^n$.
We focus on the case when $f$ is separable convex, that is, $f(\vex) = \sum_{i=1}^n f_i(x_i)$ where each $f_i: \R \to \R$ is convex.
Without this restriction, the problem is essentially solved in fixed dimension and hopelessly hard in variable dimension, as already minimizing (non-separable) convex quadratic functions over the unit cube $\{0,1\}^n$  is \NPh~\cite[Proposition 101]{EisenbrandHunkenschroederKleinKouteckyLevinOnn19}.
In the sequel, we set $\X = \Z^{n_\Z} \times \R^{n_\R}$, where $n_\Z$ and $n_\R$ should be clear from the context.
%We denote ${\displaystyle f_{\max} = \max_{\substack{\vex \in \X:\\ \vel \leq \vex \leq \veu}} |f(\vex)|}$.

\ref{MIP} is a fundamental modeling tool widely used in optimization.
As Bixby~\cite{Bixby:2002} says in his famous analysis of LP solver speed-ups, \emph{``[I]nteger programming, and most particularly the mixed-integer variant, is the dominant application of linear programming in practice.''}
Despite this, the computational complexity of~\eqref{MIP} is significantly underexplored, especially compared to the state of the art in \textsc{Integer Programming} (IP). Let us give a brief overview.

\textsc{Integer Programming} is \NPh in general, but three prominent classes of IP were shown to be tractable.
First, Hoffman and Kruskal~\cite{HoffmanK56} showed that any LP with a totally unimodular (TU) constraint matrix has integral vertices. This, together with the polynomiality of \textsc{Linear Programming} (LP), shows that ILPs with TU matrices can be solved in polynomial time.
A generalization of TU matrices are matrices with bounded subdeterminants $\Delta$; the polynomiality of such ILPs is a major open problem, so far only solved for the case $\Delta = 2$ (bimodular matrices)~\cite{ArtmannWZ17}.
Hochbaum and Shanthikumar~\cite{HochbaumShantikumar1990} extended these results to separable convex objectives.
In the '80s, Lenstra~\cite{Lenstra:1983} showed that \textsc{Integer Programming}~(IP) can be solved in time $g(n) \poly(L)$, where $L$ is the input length and $g$ is some computable function.
His algorithm works even if $f$ is a convex function.
Finally, it was recently shown that IP with separable convex $f$ can be solved in time $g(a,d)\poly(n,L)$, where $a = \|E\|_\infty$, $d$ is a measure of structural sparsity of $E$ (specifically, the smaller of the primal or dual treedepth of $E$) and $g$ is some computable function~\cite{KouteckyLO18}.
This result is a part of a major current stream of results on the complexity of IP with matrices of bounded treedepth, the state of the art being~\cite{EisenbrandHunkenschroederKleinKouteckyLevinOnn19,CslovjecsekEPVW21,CslovjecsekEHRW21,klein2021collapsing,ChanCKKP20}. IPs with bounded treedepth can intuitively be seen as IPs with small entanglement \wrt their variables and constraints, \eg, block-structured matrices such as $n$-fold or $2$-stage stochastic IPs.
Bounded treedepth IPs led to several breakthroughs in applications~\cite{KnopKM17,KnopK18,KnopKLMO:2019,JansenKMR19,ChenM18} and are the focus of a plethora of recent works.

What is the complexity of the \emph{mixed} variant~\eqref{MIP} for these classes of constraint matrices?
Already in his original paper, Lenstra showed an extension of his algorithm to~\eqref{MIP} with complexity $g(n_{\Z}) \poly(n_{\R},L)$.
But what about the other two classes? 
Let us first restrict our attention to the simpler situation where $f(\vex) = \vew\vex$ is a linear function and $\veb, \vel, \veu$ are integral, \ie,
\begin{equation}
	\min \vew \vex:\, E\vex = \veb, \, \vel \leq \vex \leq \veu, \quad \vex \in \X, \, \vel, \veu \in \Z^n, \, \veb \in \Z^m \label{MILPz} \tag{MILP$_{\Z}$}.
\end{equation}
Regarding the two aforementioned classes, the totally unimodular case is polynomial for~\eqref{MILPz} because it again coincides with LP.
For MILPs with small coefficients and treedepth, Brand, Koutecký and Ordyniak~\cite{BrandKO21} have recently shown that also this problem remains efficiently solvable.

\subsection{Our Contributions}
\subsubsection{Algorithms and Complexity of Mixed-Integer Programs.} 
The work of Hochbaum and Shanthikumar~\cite{HochbaumShantikumar1990} seemingly made it common sense that ``convex separable optimization is not much harder than linear optimization,'' as the title of their paper says.
Another case supporting Hochbaum and Shantikhumar's claim is an algorithm of Chubanov~\cite{Chubanov16}, which reduces separable convex continuous minimization to linear programming.
This nurtures the idea that positive results for the easier problem~\eqref{MILPz} translate to~\eqref{MIP}.

%outlined above nurtures the idea that positive results for~\eqref{MILPz} translate to~\eqref{MIP}.
We provide evidence against this intuition. Namely, we show that the two positive results above do \emph{not} extend to~\eqref{MIP}. This even holds if we restrict $\veb, \vel, \veu$ to be integral.
In fact, already a small generalization of~\eqref{MILPz} where we allow fractional entries, \ie,
\begin{equation}
	\min \vew \vex:\, E\vex = \veb, \, \vel \leq \vex \leq \veu, \quad \vex \in \X,\, \vel, \veu \in \X, \, \veb \in \R^m,\label{MILPf} \tag{MILP$_{\R}$}
\end{equation}
becomes much harder. We show that~\eqref{MILPf} can be reduced to~\eqref{MIP} with integral data (Lemma~\ref{lem:milp-to-mip}) by capturing the fractionality of $\veb, \vel, \veu$ by the separable convex objective $f$.
Thus, any hardness of~\eqref{MILPf} also directly shows hardness of~\eqref{MIP}.

Regarding the totally unimodular case (and, by extension, also the bimodular case), Conforti et al.~\cite{ConfortiSEW09} proved that deciding the feasibility of~\eqref{MILPf} is \NPh.
Together with our reduction to~\eqref{MIP}, this shows that~\eqref{MIP} with a TU constraint matrix is \NPh even with integral $\veb, \vel, \veu$.

In the same vein, we prove in Theorems~\ref{thm:nfoldNPh} and~\ref{thm:2stageWh1} below that for matrices of bounded treedepth, already~\eqref{MILPf}, and a fortiori the separable convex case~\eqref{MIP}, is hard.
Looking closer, we will see that the complexity of~\eqref{MILPf} for matrices of small treedepth harbours another curiosity:
As mentioned above, the prototypical examples of matrices from this class are certain block-structured matrices, namely $n$-folds and 2-stage stochastic matrices (Fig.~\ref{fig:schematicMatrices}), with blocks of size bounded by a parameter.
%\footnote{A matrix is of \emph{$n$-fold} form if non-zero entries appear only in the first few rows (that is, their number is a parameter) and in small blocks (that is, their dimensions are a parameter) along the diagonal beneath. The transpose of that matrix is called \emph{$2$-stage stochastic}.} For a schematic representation, see Figure~\ref{fig:schematicMatrices}.
\begin{figure}
\centering
\begin{subfigure}{.4\textwidth}
  \centering
\begin{equation*}
\begin{pmatrix}
\begin{tikzpicture}
  \matrix (m)[matrix of math nodes, nodes in empty cells, nodes={minimum width = 0.8cm, minimum height = 0.4cm} ]
   {
B_1	& B_2	& \dots	& B_n      \\
A_1	& 0 	& \dots  	& 0 	  \\
0	&A_2	&\dots 	& 0	\\
\vdots	& \vdots 	& \ddots & \vdots \\
0 	& 0 & \dots 	 & A_n \\
};

\draw(m-1-1.south west) rectangle (m-1-4.north east);
\draw(m-1-1.south west) rectangle (m-2-1.south east);
\draw(m-2-1.south east) rectangle (m-3-2.south east);
\draw(m-5-4.north west) rectangle (m-5-4.south east);
\end{tikzpicture}
\end{pmatrix}
\end{equation*}
\end{subfigure}
\begin{subfigure}{.4\textwidth}
  \centering
\begin{equation*}
\begin{pmatrix}
\begin{tikzpicture}
  \matrix (m)[matrix of math nodes, nodes in empty cells, nodes={minimum width = 0.8cm, minimum height = 0.4cm} ]
   {
B_1		& A_1		& 0 	& \dots 	& 0    \\
B_2		& 0 		& A_2 	& \dots  	& 0 	  \\
\vdots	& \vdots 	& \vdots & \ddots	& \vdots \\
B_n 	& 0 		& 0 & 	\dots		& A_n \\
};
\draw(m-1-1.north west) rectangle (m-4-1.south east);
\draw(m-1-1.north east) rectangle (m-1-2.south east);
\draw(m-1-2.south east) rectangle (m-2-3.south east);
\draw(m-4-5.north west) rectangle (m-4-5.south east);
\end{tikzpicture}
\end{pmatrix}
\end{equation*}
\end{subfigure}
\subcaptionbox{An $n$-fold matrix}[.4\linewidth]{}
\subcaptionbox{A $2$-stage stochastic matrix}[.4\linewidth]{}
\caption{%A schematic of $n$-fold and $2$-stage stochastic matrices. 
The $A_i$ and $B_i$ are matrices of dimension bounded by a parameter. Note that $n$-folds and $2$-stage stochastic matrices are transpositions of each other.}%of 2-stage stochastic matrices.}% with fitting dimensions to each other.}
\label{fig:schematicMatrices}
\end{figure}
While solving $n$-fold IPs can be done in single-exponential time in the parameters,\footnote{As is common in the literature, we use the term \emph{single-exponential} in $x$ for functions of the form $2^{\poly(x)}$, as opposed to e.g. $2^{O(x)}$. Similarly, we call exponential towers of height two, that is, $2^{2^{\poly(x)}}$ \emph{double-exponential} in $x$.}
%it is double-exponential for 
$2$-stage stochastic IPs require double-exponential time~\cite{JansenKL21}.
%, and this is necessary~\cite{JansenKL21}.
%The respective generalizations to tree-fold and multi-stage stochastic matrices tell a similar story:
%, \ie, where the diagonal blocks are of recursive $n$-fold or, respectively, $2$-stage stochastic form.
%Again, the former are easier to solve than the latter.
%Despite this supposed relation, 
We show that this pattern is broken 
%it does \emph{not} hold for
by~\eqref{MILPf}.
In the $n$-fold case, we show:
\begin{theorem} \label{thm:nfoldNPh}
$n$-fold~\eqref{MILPf} and~\eqref{MIP} with integral data are \NPh already with blocks of constant dimensions and with $\|E\|_\infty = 1$.
\end{theorem}
In contrast, our main algorithmic result shows that $2$-stage stochastic~\eqref{MILPf} \emph{are} solvable in polynomial time for fixed block dimensions:
% the hardness of Theorem~\ref{thm:2stageWh1} cannot be improved to \NPhness, as we give an algorithm with matching complexity:
\begin{theorem} \label{thm:2stageXP}
	$2$-stage stochastic~\eqref{MILPf} with block dimensions $r,s$ can be solved in time $g(r,s,\|E\|_\infty) n^r$ for some computable function $g$.
	\end{theorem}
The algorithm relies on a combination of insights into mixed Graver bases as well as the structure of basic solutions for 2-stage mixed-integer linear programs.
We obtain a matching lower bound for $2$-stage stochastic~\eqref{MILPf}:
\begin{theorem} \label{thm:2stageWh1}
$2$-stage stochastic~\eqref{MILPf} and~\eqref{MIP} with integral data are \W{1}-hard parameterized by the block dimensions and with $\|E\|_\infty = 1$.
\end{theorem}
Note that \Whness{1} rules out an algorithm of the form $g(k) \poly(n)$ where $k$ are the block dimensions and $n$ is the number of blocks, under the standard complexity assumption \FPT $\neq$ \W{1}.

\begin{remark}
We formulate our algorithms
% theorems 
in reference to 2-stage stochastic constraint matrices.
These are generalized by matrices of bounded primal treedepth, so-called multi-stage stochastic matrices. They are structured in much the same way as in Fig. \ref{fig:schematicMatrices}, but with diagonal blocks of recursive multi-stage stochastic form (and the depth of this recursion is bounded).
While we are confident that our algorithmic results generalize to multi-stage stochastic programs, we focus on the 2-stage case for ease of exposition.
\end{remark}

In a more general vein, we give a mixed-integer analogue of the famous algorithm of Papadimitrou~\cite{Papadimitriou81} for ILPs with few rows and small coefficients:
\begin{theorem}\label{thm:fewrows-mip}
	\eqref{MIP} can be solved in single-exponential time $(m \|E\|_\infty)^{\Oh(m^2)} \cdot \mathcal{R}$, where $\mathcal{R}$ is the time needed to solve the continuous relaxation of any~\eqref{MIP} with the constraint matrix $E$.
\end{theorem}
Until now, the best way to solve a~\eqref{MIP} with few rows and small coefficients would be to preprocess $E$ to not contain any column twice, leaving at most $(2\|E\|_\infty + 1)^m$ columns, and then use Lenstra's algorithm.
However, this leads to a double-exponential running time in terms of $m$ and single-exponential in terms of $\|E\|_\infty$.

\subsubsection*{Results on Mixed Graver Bases.}
Our algorithmic approach uses the mixed Graver basis of the constraint matrix.
This is a mixed analogue of the usual integral Graver basis, which is a central object in all the recent developments around block-structured IPs.
Deeper insights into the Graver basis have led to new dynamic data structures~\cite{EisenbrandHunkenschroederKleinKouteckyLevinOnn19}, proximity theorems~\cite{KnopKLMO:2019,CslovjecsekEHRW21,EisenbrandHunkenschroederKleinKouteckyLevinOnn19,CslovjecsekEPVW21,KleinR21}, better convergence rate analyses~\cite{EisenbrandHunkenschroederKleinKouteckyLevinOnn19}, and much more.
The mixed Graver basis was introduced by Hemmecke~\cite{Hemmecke:2000} already in 2003, but has been neglected in the literature ever since.
On our way to showing Theorem~\ref{thm:2stageXP}, we prove several results about the mixed Graver basis which are of independent interest, and disagree with the typical intuitions gained by studying the ordinary integral Graver basis.
First, while all 
%elements o $\G(E_{\text{$n$-fold}})$ 
elements of the integral Graver basis of an $n$-fold matrix have small $1$-norm, we show:
\begin{theorem} \label{thm:nfoldmip-lowerbound}
	There is an $n$-fold matrix $E$ with constant-sized blocks and $\|E\|_\infty = 1$ such that the mixed Graver basis of $E$
	%$\G_{\X}(E)$ 
	contains an element with $1$-norm $\Omega(n)$.
\end{theorem}
%\todo{How do these two theorem stand in contrast to each other? 1-norm vs. $\infty$-norm introduces a factor of $n$ anyways}

On the other hand, for $2$-stage stochastic matrices, the mixed Graver basis seems to behave similarly as the integer Graver basis, and the $\infty$-norm of its elements can be bounded by a function of the block dimensions and $\|E\|_\infty$:
\begin{theorem} \label{lem:2stagebound}
	For any $2$-stage stochastic matrix $E$, the maximum $\infty$-norm of an element of its mixed Graver basis
	% $\G_{\X}(E)$ 
	is bounded by $h(r,s,\|E\|_\infty)$ for some computable function $h$.
\end{theorem}
This bound also implies a proximity result: for any integer optimum $\vez^*$, there is a nearby mixed optimum~$\vex^*$. Thus, we can first find $\vez^*$ (which can be done efficiently), and then only search in a small neighborhood around $\vez^*$.

Until now, a bound such as $h(r,s,\|E\|_\infty)$ on the Graver elements has always lead to an algorithm with a corresponding running time $h(r,s,\|E\|_\infty) \poly(n)$.
However, in the mixed case, such an algorithm is ruled out by Theorem~\ref{thm:2stageWh1}.
This shows that in the mixed case, also the common intuition that good bounds on the Graver norm directly lead to fast algorithms is not true.

%\subsection*{Related Work}
%\cite{HochbaumShantikumar1990} \cite{BrandKO21} \cite{Hemmecke:2000} \cite{Chubanov16} \cite{ConfortiSEW09}
%
%A lot of fancy related work

\subsection{Organization}
We give all necessary preliminaries in Sect. \ref{sec:prelim}.
Then, we give new results on mixed Graver bases and algorithmic consequences for mixed-integer linear programs with few rows in Sect. \ref{sec:basic}.
In Sect. \ref{sec:2stagealg}, we then extend this to an algorithm for the 2-stage stochastic case, and Sect. \ref{sec:2stagelb} contains a matching lower bound.
In Sects. \ref{sec:nfoldlb} and \ref{sec:nfoldgraverlb}, we prove both complexity and Graver norm lower bounds for the $n$-fold case.

\section{Preliminaries} \label{sec:prelim}

We write vectors in boldface (\eg, $\vex, \vey$) and their entries in normal font (\eg, the $i$-th entry of~$\vex$ is~$x_i$).
Let $\X = \Z^{n_{\Z}} \times \R^{n_{\R}}$.
Any~\eqref{MIP} instance with infinite bounds $\vel, \veu$ can be reduced to an instance with finite bounds using standard techniques (solving the continuous relaxation and using proximity bounds to restrict the relevant region), so from now on we will assume finite bounds $\vel, \veu \in \X$.

The set of indices at which $\vex$ is non-zero is the \emph{support} of $\vex$, denoted $\suppo(\vex)$.
For positive integers $m \leq n$ we set $[m,n] \df \{m,\ldots, n\}$ and $[n] \df [1,n]$, and we extend this notation to vectors: for $\vel, \veu \in \Z^n$ with $\vel \leq \veu$, $[\vel, \veu] \df \{\vex \in \Z^n \mid \vel \leq \vex \leq \veu\}$.
If~$A$ is a matrix, $A_{i,j}$ denotes the $j$-th coordinate of the $i$-th row, $A_{i, \bullet}$ denotes the $i$-th row and $A_{\bullet, j}$ denotes the $j$-th column.
We use $\log \df \log_2$.
%For a function $f: \R^n \to \R$ and two vectors $\vel, \veu \in \R^n$, we define $f_{\max}^{[\vel, \veu]} \df \max_{\vex, \vex' \in [\vel, \veu]} |f(\vex) - f(\vex')|$; if $[\vel, \veu]$ is clear from the context, we omit it and just write $f_{\max}$.
%Moreover, we require that for each $\vex \in \Z^n$, $f(\vex) \in \Z$.
We define $\lfloor x \rceil$ to be $\floor{x}$ if $x \geq 0$ and $\ceil{x}$ otherwise, and we define the \emph{fractional part of $x$} to be $\{x\} \df x - \lfloor x \rceil$.
%We say that a system of equations $E\vex = \veb$ is \emph{pure} if the rows of $E$ are linearly independent.
The division of variables into integer and continuous ones induces a division of the constraint matrix $E = (E_{\Z} ~ E_{\R})$ where $E_{\Z} \in \Z^{m \times n_{\Z}}$ and $E_{\R} \in \R^{m \times n_{\R}}$, and analogously $\vex = (\vex_{\Z}, \vex_{\R})$ and $f(\vex) = f_{\Z}(\vex_{\Z}) + f_\R(\vex_\R)$.
More generally, whenever we make reference to any subset $E'$ of columns or even submatrix of $E$, we will freely denote with $E'_\Z$ and $E'_\R$ the analogous division of $E'$ into its integral and fractional part, respectively.
Throughout, we assume that the rows of $E$ are linearly independent.
%, as this can always be achieved by Gaussian elimination.
%\todo{def support}

We consider $n$-fold and $2$-stage stochastic matrices. A matrix is of $2$-stage stochastic structure if non-zero entries appear only in the first $r$ columns and in $n$ blocks of size $t \times s$ along the diagonal beside. The overall size is $nt \times (r+sn)$. An $n$-fold matrix is the transpose of a $2$-stage stochastic matrix. It has thus $(r+sn)$ rows and $nt$ columns. For an illustration, see Figure~\ref{fig:schematicMatrices}.  

A vector $\veg \in \ker(E) \setminus \{\vezero\}$ is a \emph{circuit of $E$} if it is integral, its entries are co-prime, and it is support-minimal, that is, there is no vector $\veg' \in \ker(E) \setminus \{\vezero\}$ with $\suppo(\veg') \subset \suppo(\veg)$; let $\CC(E)$ denote the set of circuits of $E$.
For two vectors $\vex, \vey \in \R^n$, we say that \emph{$\vex$ is conformal to $\vey$} and write $\vex \sqsubseteq \vey$ if, for each $i \in [n]$, $|x_i| \leq |y_i|$ and $x_i \cdot y_i \geq 0$.
Intuitively, $\vex$ and $\vey$ are in the same orthant, and $\vey$ is at least as far from $\vezero$ as $\vex$ in each coordinate.
We say that $\vex = \sum_i \veg_i$ is a \emph{conformal sum} or a \emph{conformal decomposition of $\vex$} if, for all $i$, $\veg_i \sqsubseteq \vex$.
For an arbitrary set $S$, we write $\ker_{S}(E)$ as a shorthand for $\ker(E) \cap S$.
In particular, the \emph{mixed kernel of $E$} is defined as $\ker_{\X}(E).$
The \emph{Graver basis of $E$}, denoted $\G(E)$, is the set $\G(E) = \{\veg \in \ker_{\Z^n}(E) \setminus \{\vezero\} \mid \veg \text{ is $\sqsubseteq$-minimal}\}$.
% = \{\vex \in \X \mid E\vex = \vezero\}$.\todo{Actually introduce the general notation $\ker_{\mathbb{A}}$ ...}
%Todo: Explanation to the definition would be nice
\begin{definition}[Mixed Graver basis \cite{Hemmecke:2000}]
Let $E = (E_\Z ~ E_\R) \in \Z^{m \times n}$.
The \emph{mixed Graver basis $\G_{\X}(E)$ of $E$ with respect to $\X$} consists of all vectors $(\mathbf{0},\veg_\R)$, where $\veg_\R \in \CC(E_\R)$, together with all vectors
$(\veg_{\Z}, \veg_{\R}) \in \ker_{\X}(E)$ such that $\veg_{\Z} \neq \vezero$ and there is no $(\veg'_{\Z}, \veg'_{\R}) \in \left(\ker_{\X}(E) \setminus \{\vezero\} \right)$  such that $(\veg'_{\Z}, \veg'_{\R}) \sqsubseteq (\veg_{\Z}, \veg_{\R})$.
\end{definition}
For any $p$, $1 \leq p \leq \infty$, define $g^{\X}_p(E) \df \max_{\veg \in \G_{\X}(E)} \|\veg\|_p$.

The following is a helpful trick to reduce a~\eqref{MILPf} to a~\eqref{MIP} with integer input data and a constraint matrix $(E~I)$.
\begin{lemma}\label{lem:milp-to-mip}
Let an~\eqref{MILPf} instance be given.
It is possible to construct an equivalent~\eqref{MIP} instance in linear time with a constraint matrix $E'=(E~I)$, bounds $\vel', \veu' \in \Z^{n+m}$, and a right-hand side $\veb' \in \Z^m$.
\end{lemma}
\begin{proof}
We will first show how to handle fractional lower and upper bounds $\vel$ and $\veu$.
Later, we will reduce the case of a fractional right-hand side $\veb$ to the case with an integer right-hand side and fractional $\vel, \veu$.
Let $M$ be a large number determined later on.
For each $i \in [n]$, define $f_i(x_i) = w_i x_i + M \cdot \dist(x_i, [l_i, u_i])$, where $\dist(a, [b,c])$ is the distance function which is $0$ if $b \leq a \leq c$ and otherwise is $\min \{|a-b|, |a-c|\}$.
Notice that $f_i$ is at most $3$-piecewise linear and convex, thus $f(\vex) = \sum_i f_i(x_i)$ is separable convex.
Let $\vel' = \floor{\vel}$ and $\veu' = \ceil{\veu}$.
Clearly, by choosing $M$ large enough, it is possible to enforce that an optimum of the new~\eqref{MIP} instance is feasible with respect to the original bounds unless the original instance is infeasible.
Moreover, notice that if $\vel \leq \vex \leq \veu$, $f(\vex) = \vew \vex$ and thus, such a solution minimizes the original objective function as desired.

To handle a fractional $\veb$, introduce $m$ slack variables $s_1, \dots, s_m$ and define the new constraint matrix as $E' = (E~I)$.
Define a new right-hand side $\veb' \df \lfloor \veb \rceil$.
Define bounds $\vel', \veu' \in \R^{n+m}$ to be $\vel$ and $\veu$ in the first $n$ variables, and for $j \in [m]$, set $l'_{n+j}, u'_{n+j} = -\{b_j\}$.
Any feasible solution $(\vex, \ves)$ of this new instance satisfies $E \vex + I \ves = E \vex - \{\veb\} = \veb' = \lfloor \veb \rceil = \veb - \{\veb\}$. Hence, by adding $\{\veb\}$ to both sides of the equation, we get that $E \vex = \veb$.
\end{proof}

\section{The Basic Case: Matrices with Few Rows and Small Coefficients}\label{sec:basic}

This section develops the basic version of our algorithmic result.
We begin by giving upper bounds for a certain notion of decompositions of elements in the mixed Graver basis, and then employ these bounds to our algorithmic ends.

\subsection{Mixed-Graver Bound}
We begin with an upper bound on the 1-norm for matrices with few rows and small coefficients.
The proof mimics that of~\cite[Lemma 27]{EisenbrandHunkenschroederKleinKouteckyLevinOnn19}; as such, we will also need the Steinitz lemma:
%Todo: I could not find the bib entries
\begin{proposition}[Steinitz~\cite{Steinitz1913}, Sevastjanov, Banaszczyk~\cite{SevastjanovBanasczyk1997}]
	\label{prop:steinitz}
	Let $\|\cdot\|$ be any norm, and let $\vex_1, \dots, \vex_n \in \R^d$ be such that $\|\vex_i\| \leq 1$ for $i \in [n]$ and $\sum_{i=1}^n \vex_i = \vezero$.
	Then there exists a permutation $\pi \in S_n$ such that for each $k \in [n]$, $\|\sum_{i=1}^k \vex_{\pi(i)}\| \leq d$.
\end{proposition}
%\todo{maybe we don't need to state this here since it is only used in the proof which is delayed?}
\begin{lemma} \label{lem:basic-mip-bound}
	Let $E \in \Z^{m \times (n_{\Z} + n_{\R})}$. Then every $\veg \in \G_{\X}(E)$ satisfies $\|\veg\|_1 \leq \left( 2\|E\|_\infty (2m\|E\|_\infty+1)^m +1\right)^m$.
\end{lemma}
\begin{proof}
	Let $\veg \in \G_{\X}(E)$ and assume that all columns of $E$ are distinct; we will show how to deal with doubled columns later.
	We define a sequence of vectors in the following manner:
	If $g_i \geq 0$, we add $\floor{g_i}$ copies of the $i$-th column of $E$ to the sequence,
	if $g_i < 0$ we add $|\ceil{g_i}|$ copies of the negation of column $i$ to the sequence.
	Thus, for each $i \in [n]$, we obtained vectors $\vev_1^i, \dots, \vev_{\lfloor g_i \rceil}^i$.
	Finally, we add the vector $\veo = \sum_{i=1}^n \{g_i\} E_{\bullet,i}$ to the sequence. Notice that this vector is integral.
	Let $q$ be the number of vectors in this sequence.
	
	Clearly, the sequence of vectors sums up to $\vezero$ as it exactly corresponds to $E \veg$ and $\veg \in \ker_{\X}(E)$.
	Moreover, their $\ell_\infty$-norm is bounded by $\|E\|_\infty (2\|E\|_\infty+1)^m$ since there are at most $(2\|E\|_\infty+1)^m$ distinct columns, $\|E\|_\infty$ is the largest number appearing in any of them, and this is an upper bound on any number appearing in $\veo = \sum_{i=1}^n \{g_i\} E_{\bullet,i}$.
	The remaining vectors $\vev^i_j$ are bounded by $\|E\|_\infty$ in $\ell_\infty$-norm.
	
	Using the Steinitz Lemma, there is a reordering $\veu^1,\dots,\veu^{q}$
	% (\ie, $\vev_j^i = \veu^{\pi(i,j)}$ for some permutation $\pi$)
	of this sequence such that
	each prefix sum $\vep_k \df \sum_{j=1}^k \veu^j$
	is bounded by $m\|E\|_\infty (2\|E\|_\infty+1)^m$ in the $l_\infty$-norm. Clearly,
	\begin{displaymath}
		\big| \{ \vex \in \Z^m \mid \|\vex\|_{\infty} \leq m \|E\|_\infty (2\|E\|_\infty+1)^m \} \big| = \left(2m\|E\|_\infty (2\|E\|_\infty+1)^m +1 \right)^m =: P.
	\end{displaymath}
	Assume for contradiction that $\|\veg\|_1 > P$. Then two of these prefix sums are the same, say, $\vep_{\alpha} = \vep_{\beta}$ with $1 \leq \alpha < \beta \leq q$.
	Obtain a vector $\veg'$ from the sequence $\veu^1, \dots, \veu^\alpha, \veu^{\beta+1}, \dots, \veu^{\|\veg\|_1}$ as follows: begin with $g'_i \df 0$ for each $i \in [n]$, and for every $\veu^\ell$ in the sequence, set
	\[g'_i \df
	\begin{cases}
		g'_i + 1 & \text{ if } \pi^{-1}(\ell) = (i,j) \text{ and } g_i \geq 0 \\
		g'_i - 1 & \text{ if } \pi^{-1}(\ell) = (i,j) \text{ and } g_i < 0 \\
		g'_i + \{g_i\} & \text{ if } \veu^\ell = \veo\text{, for each $i \in [n]$} \enspace .
	\end{cases}
	\]
	Similarly obtain $\veg''$ from the sequence $\veu^{\alpha+1} \dots, \veu^{\beta}$.
	We have $E\veg'' = \vezero$, as $\vep_\alpha - \vep_\beta = \vezero$ and thus, $\veg'' \in \ker_{\X}(E)$ and hence, $\veg' \in \ker_{\X}(E)$.
	Moreover, both $\veg'$ and $\veg''$ are non-zero and satisfy $\veg', \veg'' \sqsubseteq \veg$.
	This is a contradiction with $\sqsubseteq$-minimality of $\veg$ which is a condition needed for $\veg \in \G_{\X}(E)$, hence $\|\veg\|_1 \leq P$.
	Notice that only one of $\veg'$ or $\veg''$ may be fractional, as $\veo$ will be in exactly one subsequence.
	
	We are left to deal with the situation that $E$ contains doubled columns.
	The solution is to adjust the construction of the sequence accordingly.
	Fix a column $E_{\bullet,i}$ and let $S$ be the set of all indices $j$ such that $E_{\bullet,i} = E_{\bullet, j}$.
	Let $u = \sum_{j \in S} g_j$.
	If $u > 0$, add $\floor{u}$ copies of $E_{\bullet,i}$ into the sequence, else add $|\ceil{u}|$ copies of $-E_{\bullet,i}$ into the sequence.
	The contribution of this column type to $\veo$ will be $\{u\} E_{\bullet, i}$. Since $-1 < \{u\} < 1$ for each column type, and the number of column types is bounded by $(2\|E\|_\infty+1)^m$, our previous arguments hold.
\end{proof}
%}
The proof of the above Lemma actually shows that there exists a particular decomposition of every element of $\ker_{\X}(E)$ into an element of $\ker_{\Z^n}(E)$ (which can be further decomposed into elements of $\G(E)$) and one element of $\ker_{\X}(E)$, which we can bound.
This mixed element might not be an element of $\G_{\X}(E)$, and a bound on the elements of $\G_{\X}(E)$ does not imply a bound on this element.
We crucially need this property in our proximity bound and the bounds on $\G_{\X}(E)$ for $2$-stage matrices, as well as the prospect of extending these to multi-stage matrices.
Thus, this emerges as an important feature:
\begin{definition}[One-fat decomposition, bound]
	Let $\vex \in \ker_{\X}(E)$.
	We say that $\vex = \veh + \veg$ is a \emph{one-fat decomposition} if it is a conformal decomposition, $\veh \in \ker_{\X}(E)$ and $\veg \in \ker_{\Z^n}(E)$, and we call $\veh$ the \emph{fat element} of the decomposition.
	%Say that a one-fat decomposition is optimal if $\|\veh\|_\infty$ is minimized.\todo{needs to depend on the norm}
	For every $p$, $1 \leq p \leq \infty$, define $\wt_{p}^{\X}(\vex) = \min \|\veh\|_p$, where the minimum goes over all one-fat decompositions of $\vex$.
	Define the \emph{$\ell_p$-weight of $E$ with respect to $\X$} as $\wt^{\X}_{p}(E) = \max_{\vex \in \ker_{\X}(E)} \wt_{p}^{\X}(\vex)$.
	%\[
	%of^{\X}_\infty(E) = \max_{\substack{\vex \in \ker_{\X}(E):\\ \veh \text{ is a fat element of an optimal one-fat decomposition of } \vex}} \|\veh\|_\infty
	%\]
\end{definition}
\begin{corollary}\label{cor:fat-decomp}
	For any matrix $E$, $\wt_1^{\X}(E) \leq \left( 2m\|E\|_\infty (2\|E\|_\infty+1)^m +1\right)^m$.
\end{corollary}
\begin{proof}
	%As noted in the previous proof,
	Note that if $\vex \in \ker_{\X}(E)$ is decomposable, then it has a decomposition into $\veg', \veg''$, only one of which is fractional.
	Iterating this, we obtain the decomposition of $\vex$ into several elements of $\G_{\Z^n}(E)$, and one element of $\ker_{\X}(E)$ which is bounded as stated.
\end{proof}

We will obtain a better bound on both $g_1^{\X}(E)$ and the $\ell_1$-weight of $E$, using a recent result:
\begin{proposition}[{\cite[Lemma 1]{PaatWW20}}]
	Let $\vex^1, \dots, \vex^n \in \Z^d$ and $\alpha_1, \dots, \alpha_n \in \R_+$ such that $\sum_{i=1}^n \alpha_i \vex^i \in \Z^d$.
	If $\sum_{i=1}^n \alpha_i > d$, then there exist numbers $\beta_1, \dots, \beta_n \in \R_+$ such that, for all $i \in [n]$, $\beta_i \leq \alpha_i$ and $\sum_{i=1}^n \beta_i \leq d$, and $\sum_{i=1}^n \beta_i \vex^i \in \Z^d$.
\end{proposition}
An iterated use of this lemma gives rise to the following statement:
\begin{lemma}[Packing Lemma] \label{lem:packing}
	Let $\vex^1, \dots, \vex^n \in \Z^d$ and $\vealpha = (\alpha_1, \dots, \alpha_n) \in \R^n_+$ such that $\sum_{i=1}^n \alpha_i \vex^i \in \Z^d$.
	If $\sum_{i=1}^n \alpha_i > d$, there exist vectors $\vebeta^1, \dots, \vebeta^m \in \R^n_+$ such that, for each $j \in [m]$, $\vebeta^j \leq \vealpha$, $\sum_{i=1}^n \beta^j_i \vex^i \in \Z^d$, $\|\vebeta^j\|_1 \leq d$, and $\sum_{j=1}^m \vebeta^j = \vealpha$. Moreover, for all but at most one $j \in [m]$, $\|\vebeta^j\|_1 \geq d/2$.
\end{lemma}
\begin{proof}
	The only potentially non-obvious part is the last sentence of the statement.
	Notice that if there are $\vebeta^j$ and $\vebeta^{j'}$, $j \neq j'$, with $\|\vebeta^j\|_1, \|\vebeta^{j'}\|_1 \leq d/2$, then we can merge them.
	Formally, we set $\vebeta^{j} := \vebeta^j + \vebeta^{j'}$, and delete $\vebeta^{j'}$.
\end{proof}
Intuitively, the lemma allows us to take a non-negative linear combination of integer vectors whose result is an integer vector, and divide it into smaller such combinations while preserving the property that each smaller combination still results in an integer vector.

\begin{lemma} \label{lem:improved-basic-bound}
	Let $E \in \Z^{m \times (n_{\Z} + n_{\R})}$. 
	%Then every $\veg \in \G_{\X}(E)$ satisfies $\|\veg\|_1 \leq (2m^2\|E\|_\infty+1)^{m+1}$.
	Then $g_1^{\X}(E) \leq (2m^2\|E\|_\infty+1)^{m+1}$ and $\wt^{\X}_1 \leq (2m^2\|E\|_\infty+1)^{2m+2}$.
\end{lemma}
\begin{proof}
%	As in Lemma~\ref{lem:basic-mip-bound}, we will construct a sequence of vectors summing up to zero and %then apply the Steinitz Lemma. However, this time we will use the Packing Lemma~\ref{lem:packing} to obtain a better bound on each element of the vector sequence and thus, a better bound on the elements of $\G_{\X}(E)$ overall.
	
	Construct the vectors $\vev^i_j$ as in the proof of Lemma~\ref{lem:basic-mip-bound}.
	Now we want to get to the setting of the Packing Lemma, so we need to define a collection of vectors and a corresponding vector of coefficients.
	We are left to deal with the fractional parts of each coordinate.
	Thus, we have, for each $i \in [n]$, a non-negative coefficient $\alpha_i = |\{g_i\}|$ and a vector $\vex^i = \sign(g_i) A_{\bullet,i}$.
	%	Applying Lemma~\ref{lem:packing}, we obtain coefficient vectors $\beta^i, \dots, \beta^{\ell}$ such that, for each $j \in [\ell]$, $\sum_{i=1}^\ell \beta^j_i \vex_i$ is an integer vector and has $\infty$-norm at most $m \|E\|_\infty$.
	Applying Lemma~\ref{lem:packing}, we obtain coefficient vectors $\beta^i, \dots, \beta^{\ell}$ such that, for each $j \in [\ell]$, $\sum_{i=1}^n \beta^j_i \vex_i$ is an integer vector and has $\infty$-norm at most $m$.
	Now, notice that $\|\vealpha\|_1 \leq (2\|E\|_\infty + 1)^m$ because there are at most this many distinct columns and each coordinate of $\vealpha$ is smaller than $1$.
	Since all but at most one $\vebeta^j$ satisfy $\|\vebeta^j\|_1 \geq m/2$, we have that $\ell \leq 2\left((2\|E\|_\infty + 1)^m\right) / m +1 \leq (4 m \|E\|_\infty + 2)^m / m$.
	So, we add, for each $j \in [\ell]$, the vector $\sum_{i=1}^n \beta^j_i \vex^i$ into the sequence.
	%Todo: Why does the estimation on \ell hold?
	
	Now we use the Steinitz Lemma on the sequence $\veu^1, \dots \veu^p$.
	Since each vector in our sequence is bounded by $m \|E\|_\infty$ in $\infty$-norm, we have that unless there are at most $(2m^2 \|E\|_\infty+1)^m$ vectors in the sequence, some prefix sum repeats, the original vector $\veg$ can be decomposed, and thus every $\veg \in \G_{\X}(E)$ is bounded by $m\|E\|_\infty (2 m^2 \|E\|_\infty+1)^m$ in $1$-norm.
	However, this does \emph{not} yield a one-fat decomposition, because both the resulting vectors $\veg', \veg''$ may be mixed (non-integer).
	
	We deal with this as follows.
	Assume that some prefix sum repeats $\ell+1$ times, that is, there are indices $\gamma_1, \dots, \gamma_{\ell+1}$ such that $\vep_{\gamma_1} = \cdots = \vep_{\gamma_{\ell+1}}$.
	This means that, for each $j \in [\ell]$, the vectors in the subsequence $\veu^{\gamma_j}, \dots, \veu^{\gamma_{j+1}-1}$ sum up to $\vezero$, and a vector $\veh^j \in \ker_{\X}(E)$ can be constructed from them with the same procedure as in the proof of Lemma~\ref{lem:basic-mip-bound}.
	The same holds for the subsequence $\veu^{\gamma_{\ell+1}}, \dots, \veu^p, \veu^1, \dots, \veu^{\gamma_1 - 1}$, yielding $\veh^{\ell+1}$.
	However, because only $\ell$ of the vectors in the sequence originated by using the Packing Lemma, one of the subsequences above contains none of these vectors, and thus the corresponding vector $\veh^j$ is integral.
	This implies that if $\|\veg\|_1 > \ell \cdot (2m^2\|E\|_\infty+1)^{m}$, then $\veg$ can be decomposed into a $\ker_{\Z^n}(E)$ element and at most $\ell$ elements of $\ker_{\X}(E)$, each of them bounded by $m\|E\|_\infty (2 m^2 \|E\|_\infty+1)^m$ in $1$-norm.
	Summing all of these mixed elements up, we obtain one ``fat'' element $\veh$ with $\|\veh\|_1 \leq \ell \cdot m\|E\|_\infty (2 m^2 \|E\|_\infty+1)^m \leq (2m^2 \|E\|_\infty)^{2m+2}$.
\end{proof}

The one-fat decomposition also allows us to prove a bound on the distance between an integer and mixed optimum, which we will use in both of our algorithmic results:
\begin{lemma}[MIP Proximity] \label{lem:proximity}
	Let $\vez^* \in \Z^n$ be an integer optimum of a~\eqref{MIP} instance, and let $\vex^*$ be a mixed optimum closest to $\vez^*$ in $\ell_p$-norm, $1 \leq p \leq \infty$.
	Then $\|\vez^* - \vex^*\|_p \leq {\wt}^{\X}_p(E)$.
\end{lemma}
We will need a small technical proposition in the proof of Lemma \ref{lem:proximity}:
\begin{proposition}[{\cite[Proposition 60]{EisenbrandHunkenschroederKleinKouteckyLevinOnn19}}] \label{prop:sepconvex}
	Let $\vex, \vey_1, \vey_2 \in \R^n$, $\vey_1, \vey_2$ be from the same orthant, and $f$ be a separable convex function. Then $f(\vex + \vey_1 + \vey_2) - f(\vex + \vey_1) \geq f(\vex + \vey_2) - f(\vex)$.
\end{proposition}
\begin{proof}[Proof of Lemma~\ref{lem:proximity}]
	Assume for contradiction that $\|\vez^* - \vex^*\|_p > {\wt}^{\X}_p(E)$.
	Since $(\vez^* - \vex^*) \in \ker_{\X}(E)$, it has a one-fat decomposition $\veh + \veg$ where $\|\veh\|_p \leq {\wt}^{\X}_p(E)$.
	Because $\|\vez^* - \vex^*\|_p > {\wt}^{\X}_p(E)$, the integral part $\veg$ is non-zero.
	Let $\hat{\vez} := \vez^* - \veg = \vex^* + \veh$ and $\hat{\vex} := \vex^* + \veg = \vez^* - \veh$.
	Thus, $\vez^* - \vex^* = \veh + \veg = (\vez^* - \hat{\vex}) + (\vez^* - \hat{\vez})$.
	Now Proposition~\ref{prop:sepconvex} with $\vex = \vex^*$, $\vey_1 = \veh$, $\vey_2 = \veg$ shows
	\[
	f(\vez^*) - f(\hat{\vez}) \geq f(\hat{\vex}) - f(\vex^*) \enspace .
	\]
	By the conformality of the decomposition, $\hat{\vex}$ and $\hat{\vez}$ are within the $\vel, \veu$ bounds.
	Because $\veg \in \ker_{\Z^n}(E)$, $\hat{\vez}$ is an integer feasible solution, and because $\veh \in \ker_{\X}(E)$, $\hat{\vex}$ is a mixed feasible solution.
	Furthermore, because $\vez^*$ was an integer optimum and $\hat{\vez}$ is integer feasible, the left hand side is non-positive, and so is $f(\hat{\vex}) - f(\vex^*)$, thus $\hat{\vex}$ must be another mixed optimum and the right hand side must be zero, and so the left hand side, showing $\hat{\vez}$ to be another integer optimum.
	However, $\hat{\vez}$ is closer to $\vex^*$, a contradiction.
\end{proof}

\subsection{A Single-Exponential Algorithm}
Armed with the bounds on the mixed Graver basis and our insights into one-fat decompositions, we are now ready to develop the single-exponential algorithm.
Before we do so, however, a few general remarks are in order. These will also apply to the two-stage stochastic algorithm for fixed block-dimensions later on.
\begin{remark} \label{rem:feasibility}
Both algorithmic results will make use of the fact that if both the mixed and the integer version of the problem are feasible, then for every integral optimum, there is a mixed optimum nearby. It then suffices to first solve the (generally easier) integral version of the problem, and then solve an auxiliary mixed-integer program with the feasible region bounded by a small $n$-dimensional box around $\vex$.
Indeed, if $\vex$ is an integral solution of $E\vex = \veb, \vel \leq \vex \leq \veu$, then we will resort to solving the program $E(\vex+\vey) = \veb, \, \|\vey - \vex\|_\infty \leq P, \, \vel \leq \vex+\vey \leq \veu$ for $\vey$, which amounts to finding $\vey$ with $E\vey = \ve0, \, \vel' \leq \vey \leq \veu'$ for some new bounds $\vel',\veu'$ such that $\|\vel'-\veu'\|_\infty$ is small.
For general objectives, one optimizes the auxiliary objective $f'(\vey) = f(\vex + \vey)$, whereas for linear objectives no change is needed.
Hence, all of the algorithmic heavy lifting will be done in order to solve problems of this form.

Of course, this strategy rests on the assumption that both the mixed and the integral variant of the problem are feasible. This assumption can in turn be removed by a standard two-phase approach, similar to what is customary e.g. for the Simplex algorithm, in order to find an initial feasible solution. In short, this is done by introducing slack variables that are penalized in the objective, but admit a trivial feasible solution.
In the sequel, we will hence always assume feasibility.
\end{remark}

We say that $\vex_{\epsilon}$ is an \emph{$\epsilon$-accurate solution to~\eqref{MIP}} if there exists an optimum $\vex^*$ such that $\|\vex^* - \vex_{\epsilon}\|_\infty \leq \epsilon$. (For a discussion on the relationship of $\epsilon$-accurate and $\epsilon$-approximate optima and also the motivation to seek use the notion of $\epsilon$-accuracy, see~\cite[Section 1.2]{HochbaumShantikumar1990}.)

\begin{reptheorem}{thm:fewrows-mip}
An $\epsilon$-accurate solution of \eqref{MIP} can be found in single-exponential time $(m \|E\|_\infty)^{\Oh(m^2)} \cdot \mathcal{R}(\epsilon)$, where $\mathcal{R}(\epsilon)$ is the time needed to find an $\epsilon$-accurate optimum of the continuous relaxation of any~\eqref{MIP} with the constraint matrix $E$, and we assume $\mathcal{R}(\epsilon) \in \Omega(n)$.
\end{reptheorem}
\begin{proof}
The integer problem can be solved in time~$(m \|E\|_\infty)^{\Oh(m^2)} + \mathcal{R}(\epsilon)$ by known techniques~\cite{EisenbrandHunkenschroederKleinKouteckyLevinOnn19,EisenbrandW20} -- essentially, first solve the continuous relaxation, then reduce $\veb, \vel, \veu$ using proximity bounds, then solve a dynamic program.
Now by Lemma~\ref{lem:proximity}, a mixed optimum $\vex^*$ is at most $\wt_1^{\X}(E) \leq (2m^2\|E\|_\infty+1)^{2m+2} =: P$ far in $1$-norm.
%Let $\vex^*_{\Z}$ be the integer part of $\vex^*$.
The proximity bound implies that all prefix sums of $\vex^*_{\Z}$ with $E_{\Z}$ belong to the integer box $R := [-P, P]^{m}$, which has at most $(2P + 1)^m = (m \|E\|_\infty)^{\Oh(m^2)}$ elements.

This allows us to construct a dynamic program with $n_{\Z}+1$ stages.
Our DP table $D$ shall have an entry $D(i,\ver)$ for $i \in [n_{\Z}]$ and $\ver \in R$ whose meaning is the minimum objective attainable if the prefix sum of $\vex^*_{\Z}$ and $E_{\Z}$ restricted to the first $i$ coordinates is $\ver$.
% $\sum_{j=1}^i f^j(\veg^j_{\ver_j})$ where $\sum_{j=1}^i \ver_j = \ver$.
To that end, for all $\ver \in R$, define $x^*_i({\ver})$ to be the choice of $x^*_i \in [-P,P]$ which minimizes $f_i$ and such that $E_{\bullet,i} x^*_i = \ver$; it is possible for the solution to be undefined if no number in $[-P,P]$ satisfies the conditions.
Similarly, define $\vex^*_{\R}({\ver})$ to be an $\epsilon$-accurate minimizer of $f_{\R}$ satisfying $E_{\R} \vex^*_{\R} = \ver$.
To compute $D$, set $D(0, \ver) \df 0$ for $\ver = \vezero$ and $D(0, \ver) \df +\infty$ otherwise, and for $i \in [n_{\Z}]$, set
$$D(i, \ver) \df \min_{\substack{\ver', \ver'' \in R:\\\ver'+\ver'' = \ver}} D(i-1, \ver') + f^i(x^*_i({\ver''})) \enspace .$$
The last stage is defined as
$$D(n_{\Z}+1, \vezero) \df \min_{\substack{\ver', \ver'' \in R:\\\ver'+\ver'' = \vezero}} D(n_{\Z}, \ver') + f_{\R}(\vex^*_{\R}({\ver''})) \enspace .$$

The value of the optimal solution is $D(n_{\Z}+1, \vezero)$ and the solution $\vex^*$ itself can be computed easily with a bit more bookkeeping in the table $D$.

As for complexity, the first $n_{\Z}$ stages of the DP can be computed in time at most $n_{\Z} \cdot |R|^2 = (m \|E\|_\infty)^{\Oh(m^2)} n_{\Z}$, and the last stage solves the continuous relaxation $|R|$ times, taking time $|R| \mathcal{R}(\epsilon)$.
Altogether, the algorithm takes time at most $(m \|E\|_\infty)^{\Oh(m^2)} \mathcal{R}(\epsilon)$.
Regarding correctness, note that any $\epsilon$-accurate solution $\vex^*$ is such that $\vex^*_{\R}$ is an $\epsilon$-accurate minimizer of $E_{\R} \vex_{\R} = -E_{\Z} \vex^*_{\Z}, \, \vel_{\R} \leq \vex_{\R} \leq \veu_{\R}$, and $\vex^*_{\Z}$ is an integer minimizer of $E_{\Z} \vex_{\Z} = -E_{\R} \vex^*_{\R}, \, \vel_{\R} \leq \vex_{\R} \leq \veu_{\Z}$.
Since the algorithm finds exactly such minimizers, its correctness follows.
\end{proof}

%\todo{The algorithm for the basic case.}

\section{Upper Bounds for the 2-Stage Stochastic Case} \label{sec:2stagealg}

After giving the basic version of our algorithm for the case of few rows, we will now develop our algorithm for the case of fixed block dimension.
\begin{reptheorem}{lem:2stagebound}
	Let $E$ be a 2-stage stochastic matrix with blocks $A_1, \dots, A_n$ and $B_1, \dots, B_n$ such that each $B_i$ has at most $r$ columns and each $A_i$ has at most $s$ columns.
	%Then each $\veg \in \G_{\X}(E)$ is bounded by $h(r,s,\|E\|_\infty)$ for a double-exponential function $h$.
	Then $g^{\X}_\infty(E) \leq h(r,s,\|E\|_\infty)$ and $\wt^{\X}_\infty \leq h'(r,s,\|E\|_\infty)$ for some double-exponential functions $h, h'$.
\end{reptheorem}
To prove Theorem~\ref{lem:2stagebound}, we will need the following recent result:
\begin{proposition}[{\cite[Theorem 9]{CslovjecsekEPVW21}}] \label{prop:stronger-klein}
	Let $T_1, \dots, T_n \subseteq \Z^d$ be multisets of integer vectors of $\ell_\infty$-norm at most $\Delta$ such that their respective sums are almost the same in the following sense: there is some $\veb \in \Z^d$ and a positive integer $\epsilon$ such that
	\[
	\left\|\left(\sum_{\vev \in T_i} \vev\right) - \veb\,\right\|_\infty < \epsilon \qquad \text{for all $i \in [n]$} \enspace .
	\]
	There exists a function $f(d, \Delta) \in 2^{\Oh(d\Delta)^d}$ such that the following holds.
	Assuming $\|\veb\|_\infty > \epsilon \cdot f(d, \Delta)$, one can find non-empty submultisets $S_1,  \dots, S_n$, for all $i \in [n]$, and a vector $\veb' \in \Z^d$ satisfying $\|\veb'\|_\infty  \leq f(d, \Delta)$ such that
	\[
	\left(\sum_{\vev \in S_i} \vev\right) = \veb' \qquad \text{for all $i \in [n]$} \enspace .
	\]
\end{proposition}

\begin{proof}[Proof of Theorem~\ref{lem:2stagebound}]
	Assume $\veg = (\veg^0, \veg^1, \dots, \veg^n) \in \G_{\X}(E)$.
	For each $i \in [n]$, $(\veg^0, \veg^i) \in \ker_{\X}(B_i~A_i)$ and by 
	Lemma~\ref{lem:improved-basic-bound}, it can be decomposed as $(\veg^0, \veg^i) = (\veq^{0,i}, \veq^i) + \sum (\veh^{0,i,j}, \veh^{i,j})$ such that $\sum (\veh^{0,i,j}, \veh^{i,j})$ is integral and $\|\veq^i\|_1 \leq \wt^{\X}_1(B_i~A_i) \leq (2m^2\|E\|_\infty+1)^{2m+2} =: P$.
	Let $T_i = \{\veh^{i,j}\}_j$.
	Now, we can use Proposition~\ref{prop:stronger-klein} with $\epsilon = P$ and $\Delta = g_\infty(B_i~A_i)$, because, for each $i \in [n]$, $\veg^0 - \sum T_i = \veq^{0,i}$ and $\|\veq^{0,i}\|_\infty \leq P$.
	Proposition~\ref{prop:stronger-klein} yields that either $\|\veg^0\|_\infty \leq \epsilon \cdot f(d,\Delta) = P \cdot f(r,g_\infty(B_i~A_i)) =: h(r,s,\|E\|_\infty)$, or there are submultisets $S_1, \dots, S_n$ such that $\sum S_1 = \sum S_2 = \cdots = \sum S_n$.
	In the first case, we are done.
	In the second case, each $S_i$ corresponds to a collection of vectors $(\veh^{0,i,j}, \veh^{i,j})$ such that their sum is $(\bar{\veg}^0, \bar{\veg}^i)$ where $\bar{\veg}^0$ is indeed defined identically across $i \in [n]$ by Proposition~\ref{prop:stronger-klein}, and $\bar{\veg} \neq \vezero$ and $\bar{\veg} \sqsubseteq \veg$, a contradiction to $\veg \in \G_{\X}(E)$.
	Notice that we immediately obtain a one-fat decomposition and an $\ell_\infty$-width bound on $E$, because $\bar{\veg}$ is integral and $\veg - \bar{\veg}$ is mixed, and with iterated use of the procedure, bounded as desired.
\end{proof}

From Theorem~\ref{lem:2stagebound} and Lemma~\ref{lem:proximity}, it follows that:
\begin{corollary} \label{cor:proximity}
	Let $\vez^* \in \Z^n$ be an integer optimum of a $2$-stage stochastic~\eqref{MIP} instance.
	Then there exists a mixed optimum $\vex^* \in \X$ such that
	$\|\vez^* - \vex^*\|_\infty \leq h'(r,s,\|E\|_\infty)$ for a double exponential function $h$.
\end{corollary}

%\todo{Say something about extension to multi-stage -- highlight that we not only proved a bound on the mixed graver but also on one-fat decompositions, which we needed in the proof}

%\subsection{$2$-Stage Stochastic MILPs with Fractional Bounds}
\subsection{A Polynomial Algorithm for Fixed Block Dimension}
Using the upper bounds for $2$-stage stochastic MIPs on proximity and weight as combined in Corollary \ref{cor:proximity},
we can now formulate an algorithm which solves the $2$-stage stochastic MILP problem in polynomial time whenever the block dimensions are fixed.
We recall that $h'$ is the double-exponential function from Theorem~\ref{lem:2stagebound}.
In accordance with Remark \ref{rem:feasibility}, we note two things:
Firstly, by following a standard two-phase approach, we may assume that the problem at hand is integrally feasible. 
Then, secondly, the algorithm solves the integer program corresponding to the instance to optimality, which is fixed-parameter tractable \cite{AH,KouteckyLO18}.
We thereby obtain an integer optimum $\vez^\ast$, 
and we can now restrict ourselves to solving the following auxiliary MIP to optimality:
%We do so by elaborating on the augmentation framework which has sparked numerous results for block-structured ILPs, see \eg \cite{DBLP:conf/icalp/EisenbrandHK18,DBLP:journals/corr/abs-1904-01361,DBLP:journals/siamdm/JansenLR20,DBLP:conf/icalp/KouteckyLO18} for the latest results.

%The idea is to compute a feasible (non-optimal) solution $\vex_0$ and \emph{augment} it to optimality. 
%More precisely, let $\vey^\ast$ be an optimal solution to the \emph{augmentation MILP}:
\begin{equation}
	\min \vew \vex:\, E\vex = \vezero, \, \hat\vel \leq \vex \leq \hat{\veu}, \, \vex \in \Z^{n_{\Z}} \times \R^{n_{\R}}, \vel, \veu \in \R^n, \veb \in \Z^m. \label{AugMILP} \tag{AuxMILP}
\end{equation}
Here, $\hat{\ell}_i = \max\{\ell_i-z^\ast_i,-h'(r,s,\|E\|_\infty)\}$ and $\hat{u}_i = \min\{u_i-z^\ast_i,h'(r,s,\|E\|_\infty)\}$. 
Observe that
\begin{align} \label{eq:infty_bounds}
	\|\hat{\vel} - \hat{\veu}\|_\infty \leq 2h'(r,s,\|E\|_\infty)
\end{align}
holds.
For an optimal solution $\vex^\ast$ to \eqref{AugMILP},
the augmented solution $\vex^\ast+\vez^\ast$ is then an optimal solution to the original MILP,
by Corollary \ref{cor:proximity}.
%We say that $\vey^*$ is an optimal augmenting step. However, finding $\vey^\ast$ directly is hard, as it can be large. 
%In the all-integer case, we circumvent this by decompoding $\vey^\ast$ into few augmenting steps~$\veg_1, \dots, \veg_\ell$ of norm bounded by the Graver norm. 
%By that, and by computing an appropriate scaling factors $\lambda_1, \dots, \lambda_\ell$, we can augment our solution step-wise, \ie, 
%$\vex_i = \vex_{i-1} + \lambda_i \veg_i$ for the $i$-th iteration, $i>0$, where $\vex_{i-1}$ is the solution of the previous iteration and $\vex_\ell = \vex^*$.
%This is embedded in a binary search for the optimal step length and repeated until no strictly improving augmenting step can be found.
%%In the mixed-integer case, however, we cannot use the norm to adequately bound the number of possible augmenting steps, which requires us 
%We will first argue how to solve one of the auxiliary augmentation programs \eqref{AugMILP} to optimality, and then show how this embeds into the full augmentation framework.

%\subsection{Solving the Auxiliary Program}
What remains is to show how to solve \eqref{AugMILP} in the claimed time bound.
This is effected by  proving the following Lemma:
\begin{lemma} \label{lem:global_cont_vert}
		Let $V$ be the set of vertices of all integer slices of the auxiliary mixed-integer program \eqref{AugMILP}.
		There are at most $(8h'(r,s,\|E\|_\infty))^{(r+1)(s+1)} n^r$ 
		%global choices * block choices * local choices per block * integer non-basic local variables choices * upper/lower bound continuous non-basic local
		%r = global dim
		%$2^{rs} \cdot r \cdot \binom{n}{r} \cdot h(r,s,\|E\|_\infty)^{rs} \cdot 2^{rs}
		%=
		distinct global parts appearing in $V$, and they can be enumerated with polynomial delay.
\end{lemma}
To prove Lemma~\ref{lem:global_cont_vert}, we introduce a few bits of terminology:
If a matrix has more rows than columns, we call it \emph{portrait}.
We call it \emph{landscape} if the opposite is the case.
After arbitrarily fixing the set of integral variables $\vex_\Z$ of a mixed-integer linear program within its bounds, we are left with an ordinary linear program, that is, a polytope.
We call these polytopes the \emph{slices at $\vex_\Z$} of the mixed-integer programs.
The first $r$ columns of a $2$-stage stochastic matrix are called \emph{global}, and the remaining columns are called \emph{local}; we extend this terminology to the variables corresponding to these columns.
Similarly, if $\vex$ is a mixed solution of a $2$-stage stochastic~\eqref{MIP}, we call the first $r$ variables the \emph{global part} and the remaining variables the \emph{local part}.
We first observe:
\begin{lemma}
	The optimum of a~\eqref{MILPf} is attained at a vertex of one of its slices.
\end{lemma}
\begin{proof}
	Every optimal solution of the mixed-integer program is contained in some slice by definition.
	Even if this optimal solution is not a vertex of the slice,
	the optimum of the objective function over the slice is attained at one of it's vertices,
	since this is just an ordinary linear program. This solution is in the same slice and at least as good as the mixed-integer optimum considered before, that is, it is optimal.
\end{proof}
%Since the optimum of a linear program is attained at a vertex of the corresponding polytope,
%the optimum of a mixed-integer linear program is attained at a vertex of some of its integer slices.
%Each slice, \ie, some choice of integer variables $\vex_{\Z}$ within the given bounds, leaves us with a continuous system 
%\[
%	E_{\R} \vex_{\R} = -E_{\Z} \vex_{\Z}, \vel \leq \vex_{\R} \leq \veu.
%\]
%Note that this system is again two-stage stochastic.
%
%A vertex of the corresponding polytope is now described by a choice of basic and non-basic variables 
%$\vex_{B,\R}$ and $\vex_{N,\R}$ with the property that every coordinate of 
%$\vex_{N,\R}$ attains either its lower or upper bound and $E_{B,\R}$ is an invertible matrix.
%
%We will now describe the structure of $E_{B,\R}$.
%Recall that $E_{B,\R}$ is again two-stage stochastic, so we can speak about its blocks.
%These blocks cannot be landscape, since this would contradict $E_{B,\R}$ being of full rank.
%Since $E_{B,\R}$ is square, the entire local part of the matrix must be portrait by exactly $r_B$ rows.
%Hence, if we collect all portrait blocks into one matrix together with their global parts,
%we obtain a square matrix of size $r_B\times r_B$.
%$\vrect_E(D)$ 

Let $D$ be a subset of columns of $E$.
We fix the following notation:
\begin{itemize}
	\item $\Pi_E(D)$ is the set of blocks in $E$ such that $D$ contains at least one of its columns,
	and the corresponding block in $D$ is portrait.
	\item For each block $A \in \Pi_E(D)$, we let $\Lambda_E(A,D)$ be the set of columns of $A$ that actually appear in $D$.
	\item We write $\Gamma_E(D)$ for the set of global columns of $E$ that appear in $D$.
	%\end{itemize}
	We collect these data in a tuple $(\Pi_E(D),\{\Lambda_E(A,D)\}_{A \in \Pi_E(D)}, \Gamma_E(D))$,
	which we call the \emph{signature} of $D$ in $E$, denoted as $\Sigma_D(E)$.
	\item Given $D$, we can rearrange the variables and constraints in $E$ into a new matrix $N_D(E)$, such that (1) the global part of $N_D(E)$ contains first the variables not in $D$, then those in $D$,
	and (2) the blocks in $\Pi_E(D)$ appear as the first blocks on the diagonal,
	and within each block, first the variables in $D$ appear, and then those outside $D$. We call the square submatrix determined by the upper left corner of $N_D(E)$ and the lower right corner of the last portrait block in $N_D(E)$ the \emph{significant part} of $E$ with respect to $D$.
\end{itemize}

\begin{lemma} \label{lem:block_bound}
	Let $D$ be a subset of columns of $E$.
	If $D$ is invertible, then $|\Pi_E(D)| \leq r$.
\end{lemma}
\begin{proof}
	Observe that $D$ is again two-stage stochastic, say with $r_D \leq r$ global columns.
	Its blocks cannot be landscape, since this would contradict $D$ being invertible.
	Since $D$ is invertible, it is square.
	Therefore, the entire local part of $D$ must be portrait by exactly $r_D$ rows.
	In particular, there can be at most $r_D < r$ portrait blocks.
	%Hence, if we collect all portrait blocks of $D$ into one matrix together with their global parts in $D$,
	%we obtain a square matrix of size $r_D\times r_D$.
	%Hence, $\Pi_E(D)$ is of size $r_D \leq r.$
\end{proof}

\begin{lemma} \label{lem:signature}
	Let $\vex_\Z$ and $\vey_\Z$ be two choices for the integer variables of \eqref{AugMILP}
	such that their global parts agree.
	Let $B$ and $C$ be any two sets of continuous variables such that $E_{B,\R}$ and $E_{C,\R}$ are invertible, and assume there are vertices $\vex, \vey$ corresponding to the bases $B,C$ and agreeing with the slices $\vex_\Z,\vey_\Z$ in their integer parts, respectively.
	Suppose the following conditions hold:
	\begin{enumerate}
		\item $B$ and $C$ have the same signature in $E$.
		%They intersect the same portrait blocks: $\Pi_E(B) = \Pi_E(C)$
		%\item They contain the same global variables: $\Gamma_E(B) = \Gamma_E(C)$
		%\item In each block they intersect, they contain the same local variables: 
		%For all $A \in \Pi_E(B)$, $\Lambda_E(A,B) = \Lambda_E(A,C)$
		\item In the blocks that $B$ and $C$ intersect, $\vex$ and $\vey$ agree on all non-basic local variables and all integer local variables: For each block $A$ and each column $i$ of $A$ not in $\Lambda_E(A,B)$, 
		$\vex_i = \vey_i$.
		\item $\vex$ and $\vey$ agree in their integer and non-basic global part.
	\end{enumerate}
	Then, the entire continuous global parts of $\vex$ and $\vey$ agree.
\end{lemma}
\begin{proof}
	The system describing the slice of \eqref{AugMILP} at $\vex_\Z$ has only continuous variables $\vex_\R$ that are constrained as follows:
	\begin{align}
		E_{\R} \vex_{\R} = -E_{\Z} \vex_{\Z}, \, \hat\vel \leq \vex_{\R} \leq  \hat\veu\,.
	\end{align}
	Similarly, the slice at $\vey_\Z$ induces a system in continuous variables, say, $\vey_\R$:
	\begin{align}
		E_{\R} \vey_{\R} = -E_{\Z} \vey_{\Z}, \, \hat\vel \leq \vey_{\R} \leq  \hat\veu\,.
	\end{align}
	Since $B$ and $C$ are invertible, we have:
	\[
	\vex_{B,\R} =  -B_\R^{-1} \cdot (E_{\Z} \vex_{\Z} + {\bar B}_\R \vex_{\bar B,\R})
	\]
	and
	\[
	\vey_{C,\R} =  -C_\R^{-1} \cdot (E_{\Z} \vey_{\Z} + {\bar C}_\R \vey_{\bar C,\R}).
	\]
	Consider now the expressions
	\[
	\verho_1 = E_{\Z} \vex_{\Z} + {\bar B}_\R \vex_{\bar B,\R}
	\]
	and
	\[
	\verho_2 = E_{\Z} \vey_{\Z} + {\bar C}_\R \vey_{\bar C,\R}.
	\]
	First note that the conditions ensure that $E$ has the same significant part with respect to $B$ and $C$,
	hence we can assume $E$ to have this significant part, which we denote by $S$, in its upper left corner, of dimension, say, $d\times d$, with only zeroes to the right of it.
	Consequentially, $B^{-1}_\R$ and $C^{-1}_\R$ will have the inverse of $S$ in their respective upper left $d\times d$ corner, with only zeroes to the right of it.
	Furthermore, by assumption, $\vex$ and $\vey$ agree on the non-basic and integral portion of the variables in $S$.
	Since there are only zeroes to the right of $S$ in $E$, the first $d$ entries of $\verho_1$ and $\verho_2$ agree, which we refer to as $\verho \in \R^d$.
	But since the upper left $d\times d$ part of  $B^{-1}_\R$ and $C^{-1}_\R$ are also identical,
	so must be the first $d$ entries of $\vex_{B,\R}$ and $\vex_{C,\R}$, 
	which is the basic continuous global part of $\vex$ and $\vey$.
	Indeed, these entries are given by $-S_\R^{-1} \verho$.
	Since the integral and non-basic continuous part of $\vex$ and $\vey$ coincide by assumption,
	the entire continuous global parts of $\vex$ and $\vey$ coincide.
\end{proof}
	We can now state:
%	\begin{lemma} \label{lem:global_cont_vert}
%		Let $V$ be the set of vertices of all integer slices of the auxiliary mixed-integer program \eqref{AugMILP}.
%		There are at most $(8h'(r,s,\|E\|_\infty))^{(r+1)(s+1)} n^r$ 
%		%global choices * block choices * local choices per block * integer non-basic local variables choices * upper/lower bound continuous non-basic local
%		%r = global dim
%		%$2^{rs} \cdot r \cdot \binom{n}{r} \cdot h(r,s,\|E\|_\infty)^{rs} \cdot 2^{rs}
%		%=
%		distinct global parts appearing in $V$, and they can be enumerated with polynomial delay.
%	\end{lemma}
	\begin{proof}[Proof of Lemma~\ref{lem:global_cont_vert}]
		Let $h' := h'(r,s,\|E\|_\infty).$
		By Lemma \ref{lem:signature}, it suffices to bound the number of choices for the data that determine the global part. 
		In particular: 
		There are at most $n^r$ choices for choosing $\Pi_E(D)$, by Lemma \ref{lem:block_bound}.
		For each of the $r$ selected blocks $A$ with at most $s$ columns each in $\Pi_E(D)$, there are $2^s$ choices for $\Lambda_E(A,D)$, hence $2^{rs}$ choices for all $\Lambda_E(A,D)$ together.
		Similarly, there are $2^r$ choices for the global part $\Gamma_E(D)$.
		
		In each of the (at most) $r$ blocks of $\Pi_E(D)$, 
		the number of possible assignments to each local integral variable is bounded by $2h'$, using \eqref{eq:infty_bounds}, and there are at most two choices, $\hat u_i$ or $\hat\ell_i$, for each local non-basic continuous variable. 
		Therefore, each of the (at most) $r$ blocks with at most $s$ columns in $\Pi_E(D)$ contributes a factor of $2^s (2h')^s = (4h')^s$. All blocks collectively hence contribute a factor of $(4h')^{rs}$.
		Similarly, there are at most $r$ global integral and continuous non-basic variables to guess. 
		Again by \eqref{eq:infty_bounds}, there are only $2h'$ choices for the former variables, and two choices each for the latter. This contributes a factor of $2^r\cdot (2h')^r = (4h')^r$ to the total. 
		
		Putting this together yields the upper bound of $n^r \cdot 2^{(r+1)s} \cdot (4h')^{(s+1)r}
		\leq n^r \cdot (8h')^{(r+1)(s+1)}$.
		Clearly, all such choices can be enumerated with polynomial delay.
	\end{proof}

	Lemma \ref{lem:global_cont_vert} now suggests an obvious strategy to solve the~\eqref{AugMILP} to optimality:
	\begin{proposition} \label{prop:aux_solve}
		An optimal solution to \eqref{AugMILP} can be found in time~$h'(r,s,\|E\|_\infty)^{O(rs)}\cdot n^r$.
	\end{proposition}
	\begin{proof}
		By Lemma \ref{lem:global_cont_vert}, we may enumerate all possible global parts of vertices in the required time bound, guess the corresponding global integer values, and then solve the resulting block-diagonal mixed-integer system to optimality using the algorithm of Theorem~\ref{thm:fewrows-mip} (notice that here we are in the special case of LP which can be solved exactly, i.e., with $\epsilon = 0$, and in strongly polynomial time since $\|E\|_\infty$ is small, so $\mathcal{R}(0) = \poly(n)$).
		Among all choices of global parts, pick the one that yields the optimal value for the full program.
		This strategy is correct by Lemma \ref{lem:signature}, and runs within the required time bound by Lemma \ref{lem:global_cont_vert}.
	\end{proof}
We have now obtained:
\begin{reptheorem}{thm:2stageXP}
	$2$-stage stochastic~\eqref{MILPf} with block dimensions $r,s$ can be solved in time $k^{O(rs)} n^r$, where $k = h'(r,s,\|E\|_\infty)$ is the bound from Theorem~\eqref{lem:2stagebound}.
\end{reptheorem}	
	
	\begin{proof}
		As mentioned before, it is enough to first solve the integer program corresponding to the MILP instances, and then solving the auxiliary problem using Proposition \ref{prop:aux_solve}.
	\end{proof}
	
	\begin{remark}
		Let us note two things: Firstly, the exponent of $n$ in our algorithm is only dependent on the number $r$ of global variables. Hence, for values of $s$ such that $h'(r,s,\|E\|_\infty)^{s} \leq n^{f(r)}$ for some function $f$, our algorithm remains polynomial for fixed $r$.
		
		Secondly, note that we may choose strongly polynomial (or rather, strongly fpt) subroutines to solve the arising integer and mixed-integer programs. In this case, also the algorithm we obtain is strongly polynomial for fixed block dimensions.
	\end{remark}

\begin{comment}
\begin{theorem} \label{thm:mip-enumeration}
There is an algorithm which finds an $\epsilon$-accurate solution to a $2$-stage stochastic~\eqref{MIP} instance in time $\Oh(h'(r,s,\|E\|_\infty) 1/\epsilon)^{r} \cdot \poly(n s)$.
\end{theorem}
\begin{proof}
By proximity, we can restrict to a box of max size $2h'(r,s,\|E\|_\infty) + 1$, and the number of points of $\epsilon \Z^r$ in this box is at most $((2h'(r,s,\|E\|_\infty + 1) 1/\epsilon)^r$.
Try plugging each of them as the global part, the problem decomposes, solve each subproblem, pick the best solution.
By convexity, we must have picked a solution close enough to the optimum. \todo{Is this actually true? This might not be true ://}
\end{proof}

\todo{The enumeration algorithm for the general MIP case -- notice that it is \FPT in terms of $k, n$ but \XP in terms of $k, 1/\epsilon$.}

\todo{Because of a lower bound (forward reference), this algorithm is non-stupid. Note the lb shows one cannot have $f(k) \poly(n) \poly\log(1/\epsilon)^{g(k)}$ algorithm for any $f,g$ computable, and $\poly\log 1/\epsilon$ is a typical desired and obtained convergence rate. (However, we don't rule out an $n^{f(k)} \poly\log(1/\epsilon)$ algorithm.)}
\end{comment}

\section{W[1]-Hardness of $2$-Stage Stochastic MILPs with Fractional Bounds} \label{sec:2stagelb}

In the following we show that $2$-stage stochastic~\eqref{MILPf} and~\eqref{MIP} with integral data
is \W{1}-hard parameterized by the block dimension even if $\|E\|_\infty = 1$.
\begin{reptheorem}{thm:2stageWh1}
  $2$-stage stochastic~\eqref{MILPf} and~\eqref{MIP} with integral
  data is \W{1}-hard parameterized by the block dimensions and with
  $\|E\|_\infty = 1$.
\end{reptheorem}
\begin{proof}
  We show the theorem using a parameterized reduction from the well-known \ksum{} problem,
  which is \W{1}-hard when parameterized by the number of elements
  in a solution~\cite{DowneyF95}.

  \prob{\ksum{}}
  {A set $A$ of pairwise distinct natural numbers and two natural
    numbers $k$ and $t$.}
  {Decide whether there is a subset $S \subseteq A$ with $|S|=k$ and 
    $\sum_{s \in S}s=t$?}
	
  \noindent \textit{Transformation:}
  We give a formulation of \ksum~as a $2$-stage stochastic MILP. 
  To do so, we first scale all input numbers $a_1, a_2, \dots,
  a_n$ in $A$ and $t$ by $1/\max_i\{a_i\}$. Denote the new numbers as $a'_1, a'_2, \dots, a'_n$ and $t'$. The scaling ensures that all considered sums are smaller or equal to $1$, which comes in handy later on.
	
	Let $x_j^i$ be a binary variable that will indicate that $a'_i$ is the $j$th number appearing in the sum for all $i \in [n]$ and $j \in [k]$. We collect those numbers not appearing in a solution in a binary slack variable $x_{k+1}^i$ for each $i \in [n]$,
	yielding the constraints:
	\begin{align}
		& \sum_{j=1}^{k+1} x_j^i = 1 & \forall i \in \{1, \dots, n\}
	\end{align}
	To express the condition on the sum of the solution being $t'$, 
	we introduce fractional variables $y_j^i$ that take on the value $a'_i$ if and only if $x_j^i = 1$ for $i \in [n]$ and $j \in [k]$.
	While this is trivially achieved by $y_j^i = a_i' x_j^i$, the crux is to model this without including $a'_i$ as a coefficient, which would not be bounded by the parameter any more. 
	This is accomplished by requiring the following:
	\begin{align}
		& y_j^i \leq x_j^i & \forall i \in [n], \forall j \in [k] \\
		& \sum_{j=1}^{k+1} y_j^i = a'_i & \forall i \in [n]
	\end{align}
	This has the intended effect since $a'_i \leq 1$ by construction.
	We will then store the solution indicated by the assignment to the $x_j^i$ variables in yet another set of variables, denoted as $z_j$, where $j$ ranges from $1$ to $k$.
	\begin{align}
		& \sum_{j=1}^{k} z_j = t'
	\end{align}
	While it is easy to project the $y_j^i$ to $z_j$, the straightforward way to do so would blow up the block size to $\Omega(n).$ 
	Indeed, to ensure that the $z_j$ have the intended semantics, consider the following:
	The equality $z_j = y_j^i$ ought to be satisfied for exactly one choice of $i$, say when $i=i'$ (assuming distinct inputs); otherwise, $z_j = y_j^i + s_j^i$ holds for some non-zero compensation term $s_j^i$, whenever $i\neq i'$.
	Note that, while the $s_j^i$ do satisfy a function similar to slack variables, they may well need to be negative.
	In addition, we introduce binary variables $r_j^i$ for all $i \in [n]$ and $j \in [k]$, indicating whether or not $s_j^i = 0$. 
	The above semantics are captured in the following constraints:
	\begin{align}
		& z_j = y_j^i + s_j^i  & \forall i \in [n], \forall j \in [k] \label{eq:z_y_s_bound} \\
		& z_j \geq \min_i a'_i\\
		& -r_j^i \leq s_j^i \leq r_j^i & \forall i \in [n], \forall j \in [k] \label{eq:r_s_bound}
	\end{align}
	Our aim is then to minimize the number of times any of the $s_j^i$ are used, or conversely, to make $z_j = y_j^i$ for some $i$ as often as possible, which is expressed in the choice of objective function:
	\begin{align}
		\min  \sum_{j=1}^k \sum_{i=1}^n r_j^i
	\end{align}
	As argued, note that in a solution of a yes instance, for a fixed $j$, $z_j = y_j^i \geq \min_i a'_i$ (equivalently, $r_j^i = 0$) holds for exactly one choice of $i$, making the optimum equal to $k(n-1)$.
	%
	%Hence, for the $n-k$ choices for $i$ that do not participate in the solution, $r_j^i = 1$ holds for all $j$, contributing a value of $k$ to the objective each.
	%Similarly, the $k$ and  Clearly, the optimal value for a yes-instance is $(n-k)k + k(k-1)$.
	%That is, $(n-k)$ times we do not chose a number $a'_i$ and thus, set all of the $k$ corresponding $r_j^i$ variables to zero. However, $k$ times we chose a number $a'_i$ for some position $j$, setting the corresponding $r_j^i$ to zero and thus, only the remaining $k-1$ variables for that $i$ contribute with a positive value to the objective function.
	
	The above constraints define a $2$-stage stochastic MILP formulation with fractional variables $z_j$, $y_j^i$ and $s_j^i$, and binary variables $x_j^i$ and $r_j^i$.
	The global part is made up by the $z_j$, of which there are $k$. 
	The remaining variables are distributed across $n$ blocks of dimension $O(k)$ each, including the respective slack variables for the inequality constraints.
	The largest entry in the constraint matrix is $1 = O(k)$, and clearly, the transformation can be carried out in time polynomial in $n$ and $k$.\\
	
	\textit{Correctness:}
	We now show that an instance of \ksum~is feasible if and only if the corresponding $2$-stage stochastic MILP given by the above constraints has a solution with objective value at most $(n-1)k$.
	%(n-k)k + k(k-1)$. \\
	As noted above, by construction, a yes-instance of \ksum~leads to such a solution.
	%
	%$\Rightarrow$ Assume that the \ksum{} instance is a yes-instance. Then the constructed $2$-stage stochastic MILP has a feasible solution yielding an objective function value of $(n-k)k + k(k-1)$. \Wlogeneral, assume that the solution contains the numbers $a_1, \dots, a_k$, \ie, $\sum_{i=1}^k a_i = s$. Set the variables of the $2$-stage stochastic ILP as follows:\\
	%
	%\begin{tabular}{l l l}
	%$x_\ell^\ell = 1$ & $\forall \ell \in \{1, \dots, k\}$, \\
	%
	%$x_j^i = 0$ & $\forall i \in \{1, \dots, n\}, \forall j \in \{1, \dots, k\}, 1 \geq i \neq j \geq k$, \\
	% 
	%$x_{k+1}^i = 1$ & $\forall i \in \{k+1, \dots, n\}, 1 \geq i \neq j \geq k$, \\
	%
	%$y_\ell^\ell = a'_\ell$ & $\forall \ell \in \{1, \dots, k\}$, \\
	%
	%$y_j^i = 0$ & $\forall i \in \{1, \dots, n\}, \forall j \in \{1, \dots, k\}, 1 \geq i \neq j \geq k$, \\
	%
	%$y_{k+1}^i = a'_i$ & $\forall i \in \{k+1, \dots, n\}, 1 \geq i \neq j \geq k$, \\
	%
	%$z_j = a'_j$ &  $\forall j \in \{1, \dots, k\}$, \\
	%
	%$s_\ell^\ell = 0$ & $\forall \ell \in \{1, \dots, k\}$, \\ 
	%
	%$s_j^i = a'_j$ & $\forall i \in \{1, \dots, n\}, \forall j \in \{1, \dots, k\}, 1 \geq i \neq j \geq k$, \\
	%
	%$r_\ell^\ell = 0$ & $\forall \ell \in \{1, \dots, k\}$, \\ 
	%
	%$r_j^i = 1$ & $\forall i \in \{1, \dots, n\}, \forall j \in \{1, \dots, k\}, 1 \leq i \neq j \leq k$. \\
	%& \\
	%\end{tabular}
	%
	%It is easy to see that all constraints are satisfied and that the objective function has the desired value. \\
	
	On the other hand, assume that the MILP has a feasible solution with objective value of at most $(n-1)k.$
	Since the $r_j^i$ are binary, this implies that at least $k$ of them are zero.
	Suppose that there is an index $j$ such that $r_j^i = 1$ for all $i$.
	Then, there must be some $j'$ and  $i\neq i'$ such that $r_{j'}^i = r_{j'}^{i'} = 0$.
	By \eqref{eq:r_s_bound}, $s_{j'}^i = s_{j'}^{i'} = 0$.
	In turn, by \eqref{eq:z_y_s_bound}, $z_{j'} = y_{j'}^i = y_{j'}^{i'}.$ 
	Now, since $y_j^i$ and $y_j^{i'}$ are either zero or equal to some $a_i'$,
	and the $a'_i$ are distinct, we have that $0 = y_{j'}^i = y_{j'}^{i'} = z_j < \min_i a_i$, which contradicts the lower bound on $z_j$.
	Hence, for every $j$ there is some $i$ such that $r_j^i = 0$ holds,
	and consequentially, for all $j$, $z_j = y^{j}_i$ holds for some $i$.
	By construction, these form a valid solution to the instance of \ksum.
\end{proof}

\section{NP-hardness of $n$-Fold MIPs} \label{sec:nfoldlb}

The upper bound for 2-stage stochastic programs stands in contrast to a much stronger bound for the $n$-fold case.
Namely, we show \NPhness of $n$-fold~\eqref{MILPf} for constant parameter values.
By Lemma~\ref{lem:milp-to-mip}, we immediately get that $n$-fold~\eqref{MIP} is also \NPh for constant parameter values.
\begin{reptheorem}{thm:nfoldNPh}
An $n$-fold~\eqref{MILPf} and~\eqref{MIP} with integral data is \NPh already with blocks of constant dimensions and with $\|E\|_\infty = 1$.
\end{reptheorem}
\begin{proof}
The well-known \textsc{Partition} problem is defined as follows:
%Todo: For hardness, normally decision problems are considered
\prob{\textsc{Partition}}
{Integers $a_1, \dots, a_n$}
{Find a set $I \subseteq [n]$ such that $\sum_{i \in I} a_i = \sum_{i \not\in I} a_i$.}

Let an instance of \textsc{Partition} be given.
Without loss of generality, assume that $a_{\max} \df \max_i a_i \leq 1$; this can be achieved, \eg, by scaling every number of the original instance by $1/a_{\max}$.
We will have $n$ bricks, with brick $i \in [n]$ representing the choice whether $i \in I$ or $i \not\in I$.

Specifically, for each $i \in [n]$, introduce integer variables $x_1^i, x_2^i \in \{0,1\}$ and continuous variables $y_1^i, y_2^i$ with bounds $0 \leq y_1^i, y_2^i \leq 1$.
The local constraints (matrix $A$) are as follows:
We enforce a disjunction on the $x$-variables by the constraint
\begin{align}
x^i_1 + x^i_2 &= 1 & \forall i, \label{eq:disjunction}
\end{align}
and we enforce that $y_1^i = a_i$ iff $x_1^i = 1$ and similarly $y_2^i = a_i$ iff $x_2^i = 1$:
\begin{align}
	y^i_1 + y^i_2 &= a_1 & \forall i, \label{eq:ai} \\
	y^i_1 & \leq x_1^i & \forall i, \label{eq:y1x1} \\
	y^i_2 & \leq x_2^i & \forall i.\label{eq:y2x2}
\end{align}
Constraints~\eqref{eq:y1x1} and~\eqref{eq:y2x2} enforce that $x_1^i = 0$ implies $y_1^i=0$, as well as $x_2^i=0$ implies $y_2^i = 0$.
Combined with constraint~\eqref{eq:ai}, this enforces the desired equivalence.

It is now easy to see that the following global constraint encodes the requirement that $\sum_{i \in I} a_i = \sum_{i \not\in I} a_i$:
\begin{align}
	\sum_{i=1}^n y_1^i &= \sum_{i=1}^n y_2^i . \label{eq:partition}
\end{align}
Altogether, the instance has four variables per block, four local constraints and one global constraint, and is feasible if and only if the original \textsc{Partition} instance is.
\end{proof}

\section{Lower Bound on the Graver Norm of $n$-fold MIPs} \label{sec:nfoldgraverlb}

In this section, we will show that the $1$-norm of the mixed Graver norm can be
unbounded even for $n$-fold matrices.

We start with the following
auxiliary lemma, which is crucial for constructing an element of the
mixed Graver basis with unbounded $1$-norm.
\begin{lemma}\label{lem:nsseq}
	Let $n$ be an integer. There are two sets $S$ and $T$ of natural
	numbers with $|S|=|T|=n$ such that:
	\begin{itemize}
		\item[(1)] $\sum_{s \in S}s=\sum_{t\in T}t=2^{n^2}-1$ and
		\item[(2)] for every two subsets $S' \subseteq S$ and $T'\subseteq T$,
		with $0 < |C'\cup D'| < 2n$, 
		it holds that $\sum_{s \in S'}s\neq \sum_{t \in T'}t$.
	\end{itemize}
\end{lemma}
\begin{proof}%[Proof of Lemma~\ref{lem:nsseq}]
	Let $X \subseteq \N\setminus\{0\}$. We denote by $N(X)$, the natural number
	whose binary representation has a $1$ at the $i$-th bit
	(with $1$ being the lowest-value bit) if and only
	if $i \in X$. Conversely, for a natural number $x$, let
	$B(x)$ be the set of all indices $i$ such that the binary
	representation of $x$ is $1$ at the $i$-th bit. Note that
	$B(N(X))=X$ for every $X \subseteq \N\setminus \{0\}$.
	
	For every $i$ and $j$ with $1 \leq i,j \leq n$, let
	$p(i,j)=(i-1)n+j$.
	For every $i$ with $1 \leq i \leq n$, we set:
	\begin{itemize}
		\item $c_i$ is equal to $N(R_i)$, where $R_i=\{ p(i,j) \mid 1 \leq j \leq n\}$,
		\item $d_i$ is equal to $N(C_i)$, where $C_i=\{ p(j,i) \mid 1\leq j
		\leq n\}$.
	\end{itemize}
	We claim that setting $S=\{s_1,\dotsc,s_n\}$ and
	$T=\{t_1,\dotsc,t_n\}$ satisfies the statement of the
	lemma: As $\{B(s_1),\dotsc,B(s_n)\}$ and
	$\{B(t_1),\dotsc,B(t_n)\}$ form a partition of $[n^2]$, it holds that
	$\sum_{s \in S'}s=N(\bigcup_{s \in S'}B(s))$ and $\sum_{t \in
		T'}t=N(\bigcup_{t \in T'}B(t))$ for every subsets $S' \subseteq S$ and $T'\subseteq T$. Therefore, $\sum_{s \in
		S}s=\sum_{t\in T}t=N([n^2])=2^{n^2+1}-1$, which shows (1). 
	
	Towards showing (2), let $S'$ and $T'$ be any two subsets with $S'
	\subseteq S$ and $T'\subseteq T$ such that $0 < |S'\cup T'| < 2n$.
	Because $0 < |S'\cup T'| < 2n$, we obtain that either:
	\begin{itemize}
		\item there are $i$ and $j$ with $1 \leq i,j \leq n$ such that $s_i
		\in S\setminus S'$ and $t_j \in T'$ or
		\item there are $i$ and $j$ with $1 \leq i,j \leq n$ such that $t_i
		\in T\setminus T'$ and $s_j \in S'$.
	\end{itemize}
	Since the proofs for the two cases are analogous, we only give the
	proof for the former case.
	Let $O=B(s_i)\cap B(t_j)$ and note that $O=R_i\cap C_j=\{p(i,j)\}\neq \emptyset$.
	Since
	$t_i \in T'$, it holds that $O \in \bigcup_{t \in
		T'}B(t)$. However, due to $s_i \notin S'$, we have that
	$O \notin \bigcup_{s \in S'}B(s)$. Consequently, $\bigcup_{s
		\in S'}B(s)\neq \bigcup_{t \in T'}B(t)$ and therefore also
	$\sum_{s \in S'}s\neq \sum_{t\in T'}t$.
\end{proof}

% \begin{proof}
	%   Instead of natural numbers between $0$ and $2^{n^2}-1$,
	%   we will argue about their binary representations, that is, $0/1$-matrices of dimension $n\times n$.
	%   % to find two sets $M_1$ and $M_2$ of $n$ $0/1$-matrices each, 
	%   %such that (1) their sum is the all-ones matrix, and (2) their sets of subset sums are disjoint.
	
	%   Let $r_{i}\in\{0,1\}^{n\times n}$ have ones in the $i$-th row, zeros everywhere else,
	%   and let $c_i = r_i^T$.
	%   Choose ${R} = \{r_1,\ldots,r_n\}$ and ${C} = \{c_1,\ldots,c_n\}$,
	%   and let $S$ and $T$ be the sets of natural numbers that the matrices in $R$ and $C$ correspond to as binary representations, respectively.
	
	%   The fact that both $R$ and $C$ sum up to the all-ones matrix implies
	%   that $S$ and $T$ each satisfy property (1).
	%   Now, observe that the sum over every proper subset of $R$ has no non-zero columns.
	%   However, every sum over a proper subset of $C$ will have at least one all-zeros column.
	%   Hence, these two sums aren't equal.
	%   Since the elements of $R$ (and hence $C$) have pairwise disjoint support,
	%   this implies property (2) for $S$ and $T$.
	% \end{proof}
\begin{reptheorem}{thm:nfoldmip-lowerbound}
  Let $n$ be an integer, $\X_n=(\Z\times \R\times \R)^n$, and 
  $E_n$ be the matrix given by the $n$-fold of $\begin{pmatrix}
    0 & I_3 \\ 0 & A \end{pmatrix}$, where $I_3$ is the identity
  matrix of dimension $3$ and $A=(1,1,1)$. Then,
  $\G_{\X_n}(E_n)$ contains a vector $\veg$ with $\|\veg\|_1\in
  \Omega(n)$, i.e., $g_1^{\X_n}(E_n) \in \Omega(n)$.
\end{reptheorem}
\begin{proof}
	Without loss of generality we assume that $n=2m$ for some integer $m$.
	Let $S=\{s_1,\dotsc,s_m\}$ and $T=\{t_1,\dotsc,t_m\}$ be the sets
	with $m$ elements each, whose existence is guaranteed due to
	Lemma~\ref{lem:nsseq}. Let $V=2^{n^2+1}-1$.
	Let $\veg \in \X_n$ be defined as follows:
	\begin{itemize}
		\item for every $i$ with $1 \leq i \leq m$, we set:
		$g_{3(i-1)}=-1$, $g_{3(i-1)+1}=s_i/V$, and
		$g_{3(i-1)+2}=1-s_i/V$,
		\item for every $i$ with $m+1 \leq i \leq 2m$, we set:
		$g_{3(i-1)}=1$, $g_{3(i-1)+1}=-t_i/V$, and
		$g_{3(i-1)+2}=-1+t_i/V$.
	\end{itemize}
	
	We claim that $\veg \in
	\G_{\X_n}(E_n)$; because $\|\veg\|_1\in \Omega(n)$ this would prove the theorem. We start by showing that
	$\veg \in \ker_{\X_n}(E_n)$, \ie, $\veg \in \X_n$ and $E_n\veg=0$.
	By definition of $\veg$, $\veg \in \X_n$ and moreover, the following shows that
	the first 3 rows of $E_n\veg$ are equal to $0$.
        The first row is equal to:
        \[\sum_{i=0}^{n-1}g_{3i}=(\sum_{i=1}^m-1)+(\sum_{i=1}^m1)=0.\]
        
        The second row is equal to (using (1) of Lemma~\ref{lem:nsseq}): 
	\[\sum_{i=0}^{n-1}g_{3i+1}=(\sum_{i=1}^ms_i/V)+(\sum_{i=1}^m-t_i/V)=(\sum_{i=1}^ms_i-t_i)/V=0.\]
        
        Finally, the third row is equal to (using (1) of Lemma~\ref{lem:nsseq}): 
        \[\sum_{i=0}^{n-1}g_{3i+2}=(\sum_{i=1}^m1-s_i/V)+(\sum_{i=1}^m-1+t_i/V)=(\sum_{i=1}^m-s_i+t_i)/V=0.\]
		
	Moreover, for every $i$ with $1 \leq i \leq m$, we obtain the value
	of the $(3+i)$-th row as:
	
	\begin{align*}
		g_{3(i-1)}+g_{3(i-1)+1}+g_{3(i-1)+2}=-1+s_i/V+1-s_i/V=0.
	\end{align*}
	
	Similarly, for every $i$ with $m+1 \leq i \leq 2m$, we obtain the value
	of the $(3+i)$-th row as:
	
	\begin{align*}
		g_{3(i-1)}+g_{3(i-1)+1}+g_{3(i-1)+2}=1-t_i/V-1+t_i/V=0.
	\end{align*}
	
	Therefore, $\veg \in \ker_{\X_n}(E_n)$. It remains to show that $\veg$ is
	$\sqsubseteq$-minimal w.r.t. all vectors in $\ker_{\X_n}(E_n)\setminus\{\vezero\}$.
	
	Towards showing this, we first show that the restriction
	$\veg^i=(g_{3(i-1)},g_{3(i-1)+1},g_{3(i-1)+2})$ of $\veg$ to
	the $i$-th block of $E_n$ is $\sqsubseteq$-minimal w.r.t. all
	vectors in $\ker_{\X_1}(A)\setminus \{\vezero\}$, where $\X_1=\Z\times\R\times\R$; note
	that $\veg^i \in \ker_{\X_1}(A)$ because $\veg \in \ker_{\X}(E_n)$.
	Note also that either $\veg^i=(-1,s/V,1-s/V)$
	(if $i\leq m-1$) or $\veg^i=(1,-t/V,-1+t/V)$ (if $i\geq m$) for some
	$s \in S$ respectively $t\in T$.
	Suppose for a contradiction that this is not the case and there is a
	vector $\vey \in \ker_{\X_1}(A)\setminus\{\vezero\}$ such that $\vey
	\sqsubseteq \veg^i$. 
	Moreover, note that because $\vey
	\in \ker_{\X_1}(A)\setminus\{\vezero\}$,
	it also holds that
	$y_0+y_1+y_2=0$ but $\vey\notin \{\vezero,\veg\}$.
	
	We start by showing the claim for the case that
	$\veg^i=(-1,s/V,1-s/V)$.
	Because $\vey \sqsubseteq \veg^i$ and $\vey \in \X_1$,
	we obtain that $y_0\in \{-1,0\}$, $y_1\in
	[0,s_i/V]$, and $y_2\in [0,1-s_i/V]$.
	Therefore, if $y_0=0$, then $y_1+y_2=0$, which implies
	that $\vey=\vezero$, a contradiction.
	Similarly, if
	$y_0=-1$, then $y_1+y_2=1$, which is only possible if
	$\vey=\veg^i$, a contradiction.
	
	The proof for the case that $\veg^i=(1,-t/V,-1+t/V)$ is now
	analogous. 
	In particular,
	because $\vey \sqsubseteq \veg^i$  and $\vey \in \X_1$,
	we obtain that $y_0\in \{0,1\}$, $y_1\in
	[-t_i/V, 0]$, and $y_2\in [-1+t_i/V,0]$.
	Therefore, if $y_0=0$, then $y_1+y_2=0$, which implies
	that $\vey=\vezero$, a contradiction.
	Similarly, if
	$y_0=1$, then $y_1+y_2=-1$, which is only possible if
	$\vey=\veg^i$, a contradiction.
	This completes the proof that $\veg^i$ is $\sqsubseteq$-minimal w.r.t. all
	vectors in $\ker_{\X_1}(A)\setminus \{\vezero\}$.
	
	Towards showing that $\veg$ is $\sqsubseteq$-minimal w.r.t. all
	vectors in $\ker_{\X}(E_n)\setminus \{\vezero\}$, suppose for a
	contradiction that this is not the case and there is a vector $\vey\in
	\ker_{\X}(E_n)\setminus \{\vezero\}$ with $\vey \sqsubseteq \veg$.
	Note that the fact that $\veg^i$ is $\sqsubseteq$-minimal w.r.t. all
	vectors in $\ker_{\X_1}(A)\setminus \{\vezero\}$ implies that the
	restriction $\vey^i=(y_{3(i-1)},y_{3(i-1)+1},y_{3(i-1)+2})$ of
	$\vey$ to the $i$-th block of $E_n$ is either equal to $\vezero$ or
	equal to $\veg^i$. Therefore, if $\vey\notin \{\vezero,\veg\}$, then the value of
	the second row of $E_n\vey$, i.e., 
	$\sum_{i=0}^{n-1}y_{3i+1}$, is equal to $(\sum_{s \in
		S'}s/V)+(\sum_{t\in T'}-t/V)=(\sum_{s \in
		S'}s+\sum_{t\in T'}-t)/V$ for some subsets $S' \subseteq S$
	and $T'\subseteq T$ such that $0<|S'\cup T'|<2n$.
	Because of Lemma~\ref{lem:nsseq} (2), we obtain that
	$\sum_{s \in    S'}s\neq\sum_{t\in T'}t$ and therefore the value of
	the second row of $E_n\vey$ cannot be $0$, which contradicts our
	assumption that $\vey \in \ker_\X(E_n)$.
\end{proof}

\section{Takeaways and Open Problems}
We have shown that several popular intuitions do not transfer either from the linear to the separable convex case or from the integer to the mixed integer case, or both.
What, then, are the intuitions one can take away from our results?

One is the contrast between the complexities of~\eqref{MILPz} and~\eqref{MILPf}.
We take this to mean that the source of hardness for MILP is not really in allowing continuous \emph{domains}, but in the combination of continuous domains and (highly) fractional \emph{input data}, be it in the explicit form of the bounds and right hand side, or the (potentially) implicit form of the objective function.
Simply put, we still find it reasonable to expect that~\eqref{MILPz} is not much harder than ILP with the same constraint matrix, given that $\veb, \vel, \veu$ is integral.
From this perspective, the complexity of bimodular~\eqref{MILPz} is an interesting open problem.

Other takeaways come from the interplay of mixed Graver norms and algorithms based on them.
One reason even good mixed Graver norms are not as helpful in the mixed case is that they do not allow simple enumeration, as in the integer case; bounding few variables in a small (continuous) box is no help without additional structure.
The structure seems to be this: in $2$-stage stochastic MILPs, only few continuous variables interact at a time.
This is in contrast with $n$-fold MILPs which may encode the interaction of many continuous variables simultaneously.
Finally, the intuition that good Graver bounds yield good proximity bounds needs an adjustment: for good proximity, one needs not only a mixed Graver bound, but the (stronger) weight bound.
Still, in the settings of interest to us, those two come hand in hand.
It is unclear whether this is a general feature, and we pose it as a question: can $\wt^{\X}_p(E)$ be bounded in terms of $g_p^{\X}(E)$?

%\bibliography{milp}

\newcommand{\etalchar}[1]{$^{#1}$}

\end{document}